\documentclass[11pt, letterpaper, fleqn]{article}
\usepackage{verbatim,amssymb,amsthm,amsfonts, color, cite}
\usepackage{amsmath, times, amsthm}
\usepackage{amsmath, amssymb, amsthm, mathtools, tikz, todonotes, bbm, hyperref, fullpage}
\usepackage[noblocks]{authblk}
\usepackage{dsfont}

\newcommand\R{\mathbb R}

\newcommand{\adeg}{\widetilde{\mathrm{deg}}}
\newcommand{\eps}{\varepsilon}

\renewcommand{\H}{\mathcal{H}}

\newcommand{\bra}[1]{\{#1\}}
\newcommand{\pmone}{\bra{-1, 1}}

\renewcommand{\deg}{\mathrm{deg}}
\newcommand{\supp}{\mathrm{supp}}

\newcommand{\mathify}[1]{\ifmmode{#1}\else\mbox{$#1$}\fi}
\newcommand{\abs}[1]{\mathify{\left| #1 \right|}}

\newcommand{\AND}{\mathsf{AND}}
\newcommand{\OR}{\mathsf{OR}}
\newcommand{\EQ}{\mathsf{EQ}}

\def\B{\{0,1\}}

\providecommand\abs[1]{\lvert#1\rvert}
\providecommand\bigabs[1]{\bigl\lvert#1\bigr\rvert}
\providecommand\ip[1]{\langle#1\rangle}

\newtheorem{theorem}{Theorem}[section]
\newtheorem{corollary}[theorem]{Corollary}
\newtheorem{lemma}[theorem]{Lemma}

\newtheorem{claim}[theorem]{Claim}
\newtheorem{fact}[theorem]{Fact}

\newtheorem{remark}[theorem]{Remark}

\newtheorem*{theorem*}{Theorem}
\newtheorem*{corollary*}{Corollary}

\DeclareMathOperator\pr{\mathrm{Pr}}
\DeclareMathOperator\E{\mathbb{E}}

\def\eps{\varepsilon}

\newcommand{\poly}{\text{poly}}

\title{Approximate degree, secret sharing, and concentration phenomena}

\date{}

\begin{document}
\author[1]{Andrej Bogdanov}
\author[2]{Nikhil S. Mande}
\author[2]{Justin Thaler}
\author[1]{Christopher Williamson}
\affil[ ]{\textit {\{andrejb, chris\}@cse.cuhk.edu.hk}}
\affil[ ]{\textit {\{nikhil.mande, justin.thaler\}@georgetown.edu}}

\affil[1]{Chinese University of Hong Kong}
\affil[2]{Georgetown University}

\maketitle

\begin{abstract}
The $\eps$-approximate degree $\widetilde{\text{deg}}_\eps(f)$ of a Boolean function $f$ is the least degree of a real-valued polynomial that approximates $f$ pointwise to within $\eps$.  A sound and complete certificate for approximate degree being at least $k$ is a pair of probability distributions, also known as a \emph{dual polynomial}, that are perfectly $k$-wise indistinguishable, but are distinguishable by $f$ with advantage $1 - \eps$.  Our contributions are:

\begin{itemize}
\item We give a simple, explicit new construction of a dual polynomial for the $\AND$ function on $n$ bits, certifying that its $\eps$-approximate degree is $\Omega\left(\sqrt{n \log 1/\eps}\right)$.  This construction is the first to extend to the notion of weighted degree, and yields the first explicit certificate that the $1/3$-approximate degree of any (possibly unbalanced) read-once DNF is $\Omega(\sqrt{n})$. It draws a novel connection between the approximate degree of $\AND$ and
anti-concentration of the Binomial distribution.

\medskip

\item We show that any pair of {\em symmetric} distributions on $n$-bit strings that are perfectly $k$-wise indistinguishable are also statistically $K$-wise indistinguishable with at most $K^{3/2} \cdot \exp\left(-\Omega\left(k^2/K\right)\right)$ error for all $k < K \leq n/64$.
This bound is essentially tight, and implies that any symmetric function $f$ is a reconstruction function with constant advantage for a ramp secret sharing scheme that is secure against size-$K$ coalitions with statistical error $K^{3/2} \cdot \exp\left(-\Omega\left(\widetilde{\text{deg}}_{1/3}(f)^2/K\right)\right)$ for all values of $K$ up to $n/64$ simultaneously.
Previous secret sharing schemes required that $K$ be determined in advance, and only worked for $f=\AND$.  Our analysis draws another new connection between approximate degree and concentration phenomena.

\medskip
As a corollary of this result, we show that for any $d \leq n/64$, any degree $d$ polynomial approximating a symmetric function $f$ to error $1/3$
must have coefficients of $\ell_1$-norm at least $K^{-3/2} \cdot \exp\left({\Omega\left(\widetilde{\text{deg}}_{1/3}\left(f\right)^2/d\right)}\right)$.
We also show this bound is essentially tight for any $d > \widetilde{\text{deg}}_{1/3}(f)$.
These upper and lower bounds were also previously only known in the case $f=\AND$.
\end{itemize}

\thispagestyle{empty}
\addtocounter{page}{-1}
\newpage

\end{abstract}
\clearpage
\hypersetup{pageanchor=true}
\setcounter{page}{1}

\section{Introduction}

The $\varepsilon$-approximate degree of a function $f\colon \{-1, 1\}^n \to \{0, 1\}$, denoted $\adeg_{\varepsilon}(f)$, is the least degree of a multivariate real-valued polynomial $p$ such that $|p(x)-f(x)| \leq \varepsilon$ for all inputs $x \in \{-1,1\}^n$.\footnote{In this work, for convenience we also
consider functions mapping $\{0, 1\}^n$ to $\{0, 1\}$.}  Such a $p$ is said to be an approximating polynomial for $f$.  This is a central object 
of study in computational complexity, owing to its polynomial equivalence to many other complexity measures including sensitivity, exact degree, deterministic and randomized query complexity~\cite{NiS94}, and quantum query complexity~\cite{BCdWZ99}.

By linear programming duality, $f$ has $\eps$-approximate degree more than $k$ if and only if there exist a pair of probability distributions $\mu$ and $\nu$ over the domain of $f$ such that $\mu$ and $\nu$ are perfectly $k$-wise indistinguishable (i.e., all $k$-wise projections of $\mu$ and $\nu$ are identical), but are $(1 - \eps)$-distinguishable by $f$, namely $\E_{X \sim \mu}[f(X)] - \E_{Y \sim \nu}[f(Y)] \geq 1 - \eps$.  Said equivalently, a sound and complete certificate for $\varepsilon$-approximate degree being more than $k$ is a {\em dual polynomial} $q = (\mu - \nu)/2$ that contains no monomials of degree $k$ or less, and such that $\sum_x \abs{q(x)} = 1$ and $\sum_x q(x) f(x) \geq \eps$. 

Dual polynomials have immediate applications to cryptographic secret sharing:
a dual polynomial $q = (\mu - \nu)/2$ for $f$ is a description of a 
cryptographic scheme for sharing a 1-bit secret amongst $n$ parties, where the secret can be reconstructed by applying $f$ to the shares, and 
the scheme is secure against coalitions of size $k$ (see \cite{BIVW} for details). 

\medskip
\noindent \textbf{Motivation for explicit constructions of dual polynomials.} 
Recent years have seen significant progress in proving new approximate degree lower bounds by explicitly constructing dual polynomials exhibiting
the lower bound \cite{BunT13, SheANDOR, BT13, sherstov2018breaking, BunT17, BKT18, BT18, SheWu18}. These new lower bounds have in turn resolved significant open questions
in quantum query complexity and communication complexity. 
At the technical core of these results are techniques for constructing a dual polynomial for composed functions $f \circ g := f(g, \dots, g)$, given dual polynomials
for $f$ and $g$ individually.

Often, an explicitly constructed dual polynomial showing that $\adeg_{\varepsilon}(g)\geq d$  
exhibits additional metric properties, beyond what is required simply to
witness $\adeg_{\varepsilon}(g)\geq d$. 
Much of the major recent progress in proving approximate degree lower bounds has exploited 
these additional metric properties \cite{BunT17, BKT18, BT18, SheWu18}. Accordingly, 
even if cases where an approximate degree lower bound for a function $g$ is known, it can often
be useful to construct an explicit dual polynomial witnessing the lower bound.
Hence, we are optimistic that the new constructions of
dual polynomials given in this work will find future applications.

Explicit constructions of dual polynomials
are also necessary to implement the corresponding secret-sharing scheme, and to analyze the complexity of
the algorithm that samples the shares of the secret.

\medskip
\noindent \textbf{Our results in a nutshell.} Our results fall into two categories. In the first category, we reprove several known approximate degree lower bounds
by giving the first explicit constructions of dual polynomials witnessing the lower bounds. 
Specifically, our dual polynomial 
certifies that the $\eps$-approximate degree of the $n$-bit $\AND$ function is $\Theta(\sqrt{n \log 1/\eps})$.  This construction is the first to extend to the notion of weighted degree, and yields the first explicit certificate that the $1/3$-approximate degree of any (possibly unbalanced) read-once DNF is $\Omega(\sqrt{n})$.
Interestingly, our dual polynomial construction draws a novel and clean connection between the approximate degree of $\AND$ and
anti-concentration of the Binomial distribution. 

In the second category, we prove new and tight results about
the size of the coefficients of polynomials that approximate symmetric functions.
Specifically, we show that for any $d \leq n/64$, any degree $d$ polynomial approximating
$f$ to error $1/3$ must have coefficients of weight ($\ell_1$-norm) at least 
$d^{3/2} \cdot \exp\left({\Omega\left(\widetilde{\text{deg}}_{1/3}\left(f\right)^2/d\right)}\right)$.
We show this bound is tight (up to logarithmic factors in the exponent) for any $d > \widetilde{\text{deg}}_{1/3}(f)$. These bounds were previously only known in the case $f=\AND$ \cite{STT, BW17}.
Our analysis actually establishes a considerably more general result, and as a consequence we obtain new cryptographic secret sharing
schemes with symmetric reconstruction procedures (see Section \ref{s:details} for details).

\vspace{-3mm}
\subsection{A New Dual Polynomial for $\AND$}

To describe our dual polynomial for $\AND$, it will be convenient 
to consider the $\AND$ function to have domain $\{-1, 1\}^n$ and range $\{0, 1\}$,
with $\AND(x)=1$ if and only if $x=1^n$. 
In their seminal work, Nisan and Szegedy~\cite{NiS94} proved that the $1/3$-approximate degree of the $\AND$ function on $n$ inputs is $\Theta(\sqrt{n})$.
More generally, it is now well-known that the $\eps$-approximate degree of $\AND$ is
$\Theta\left(\sqrt{n \log(1/\eps)}\right)$ \cite{KahnLS96, BCdWZ99}.
These works do not construct explicit dual polynomials witnessing the lower bounds; this was achieved later in works of \v{S}palek~\cite{Spalek08} and Bun and Thaler~\cite{BunT13}.

Our first contribution is the construction of a new dual polynomial $\phi$ for $\AND$, which is simple enough to describe in a single equation:
\begin{equation}
\label{eq:dualq}
 \phi(x) = \frac{(-1)^n}{Z} \biggl(\prod_{i \in [n]} x_i\biggr) \biggl(\E_S\prod_{i \in S} x_i \biggr)^2. \end{equation}
Here, $S$ is a random subset of $\{1, \dots, n\}$ of size at most $\tfrac12 (n - d)$ (where $d$ determines
the degree of the polynomials against which the exhibited lower bound holds), and $Z$ is an (explicit) normalization constant.  

In the language of secret sharing, to share a secret $s \in \{-1, 1\}$, the dealer samples shares $x \in \{-1, 1\}^n$ with probability proportional to $(\E_S\prod_{i \in S} x_i)^2$, conditioned on the parity of the shares $\prod x_i$ being equal to $s$.

In Corollary~\ref{cor:and} we show that $\phi$ certifies that every degree-$d$ polynomial must differ from the $\AND$ function by $2^{-n} \sum_{k = 0}^{(n - d)/2} \binom{n}{k}$ at some input.  In other words, the approximation error of a degree-$d$ polynomial is lower bounded by the probability that a sum of unbiased independent bits deviates from its mean by $d/2$.

Our function $\phi$ given in \eqref{eq:dualq}, unlike previous dual polynomials~\cite{KahnLS96, Spalek08, BT13, sherstov15}, also certifies that the \emph{weighted} $1/3$-approximate degree of $\AND$  with weights $w \in \R_{\geq 0}^n$ is $\Omega(\|w\|_2)$ (see Corollary~\ref{cor:weighted}).\footnote{
For a polynomial $p(x_1, \dots, x_n)$, a weight vector $w \in \R_{\geq 0}^n$ assigns weight $w_i$ to variable $x_i$. The weighted degree of $p$ is the maximum weight over all monomials appearing in $p$,
where the weight of a monomial is the sum of the weights of the variables appearing within it. The weighted $\eps$-approximate degree of $f$,
denoted $\adeg_{w, \eps}(f)$, is the least weighted degree of any polynomial that approximates $f$ pointwise to error $\eps$.}
This lower bound is tight for all $w$, matching an upper bound of Ambainis \cite{Amb10}.   The only difference in our dual polynomial construction for the weighted case is in the distribution over sets $S$, and the lower bound in the weighted case is derived from anti-concentration of \emph{weighted} sums of Bernoulli random variables.  

Both statements are corollaries of the following theorem.
\begin{theorem}
\label{thm:dualpoly}
Define $\AND \colon \{-1, 1\}^n \to \{0, 1\}$ as $\AND(x)=1$ if and only if $x=1^n$. 
The function $\phi$ defined in Equation \eqref{eq:dualq} is a dual witness for $\adeg_{w, \eps}(\AND) \geq d$ for 
$\eps = \pr_{X \sim \{-1, 1\}^n}[\ip{w, X} \geq d]$. 
\end{theorem}

By combining, in a black-box manner, 
the dual polynomial for the weighted-approximate
degree of $\AND$ with prior work (e.g., \cite[Proof of Theorem 7]{KV}), one obtains, for any read-once DNF $f$, an explicit
dual polynomial for the fact that $\adeg_{1/3}(f) \geq \Omega(n^{1/2})$. 
Very recent work of Ben-David et al.~\cite{BBGK18} established this result for the first time, shaving
logarithmic factors off of prior work \cite{BT13, KV}. In fact,  Ben-David et al. \cite{BBGK18} prove more generally that any depth-$d$ read-once $\AND$-$\OR$ formula has approximate degree~$2^{-O(d)} \sqrt{n}$.  Their method, however, does not appear to yield an explicit dual polynomial, even in the case $d=2$.

\medskip \noindent \textbf{Discussion.}
It has been well known that the $\eps$-approximate degree of the $\AND$ function on $n$ variables is $\Theta\left(\sqrt{n \log (1/\eps)}\right)$ \cite{NiS94, BCdWZ99}, a fact which has many applications in theoretical computer science. This is superficially reminiscent of Chernoff bounds, which state that the middle $\Theta\left(\sqrt{n \log(1/\varepsilon)}\right)$ layers of the Hamming cube contain a $1-\varepsilon$ fraction of all inputs (i.e., ``most'' $n$-bit strings have Hamming weight close to $n/2$). However, these two phenomena have not previously been connected, and it is not a priori clear why approximate degree should be related to concentration of measure. An approximating polynomial $p$ for $f$ must approximate $f$ at \emph{all} inputs in $\{-1, 1\}^n$. Why should it matter that \emph{most} (but very far from all) inputs have Hamming weight close to $n/2$? 

The new dual witness for $\AND$ constructed in Equation \eqref{eq:dualq} above provides a surprising answer to this question. The connection between (anti-)concentration and approximate degree of $\AND$ arises not because of the number of \emph{inputs} to $f$ that have Hamming weight close to $n/2$, but because of the number of \emph{parity functions} on $n$ bits that have \emph{degree} close to $n/2$. This connection appears to be rather deep, as evidenced by our construction's ability to yield a tight lower bound in the case of weighted approximate degree.

\vspace{-1.85mm}
\subsection{Indistinguishability for Symmetric Distributions}
\label{s:details}
In this section, for convenience we consider functions mapping $\{0, 1\}^n$ to $\{0, 1\}$. Two distributions $\mu$ and $\nu$ over $\{0,1\}^n$ are {\em (statistically) $(k, \delta)$-wise indistinguishable} if for all subsets $S \subseteq \{1, \dots, n\}$ of size $k$, the induced marginal distributions $\mu|_S$ and $\nu|_S$ are within statistical distance $\delta$.   When $\delta = 0$, we say they are {\em (perfectly) $k$-wise indistinguishable}.

For general pairs of distributions, perfect $k$-wise indistinguishability does not imply any sort of security against distinguishers of size $k + 1$.  Any binary linear error-correcting code of distance $k + 1$ and block length $n$ induces a pair of distributions (the uniform distribution over codewords and one of its affine shifts) that are perfectly $k$-wise indistinguishable, yet perfectly $(k + 1)$-wise distinguishable.  

In contrast, we prove that perfect $k$-wise indistinguishability for {\em symmetric} distributions implies strong statistical security against larger adversaries:  

\begin{theorem}
\label{thm:mainupper}
If $\mu$ and $\nu$ are symmetric over $\{0, 1\}^n$ and perfectly $k$-wise indistinguishable, then they are statistically $(K, O(K^{3/2}) \cdot e^{-k^2/1156K})$-wise indistinguishable for all $1\leq k < K \leq n/64$.
\end{theorem}

Theorem~\ref{thm:mainupper} has the following direct consequence for secret sharing schemes over bits with symmetric reconstruction.  We say $(\mu, \nu)$ are $\alpha$-reconstructible by $f$ if $\E_{X \sim \mu}[f(X)] - \E_{Y \sim \nu}[f(Y)] \geq \alpha$.

\begin{corollary}
\label{cor:symmetricimperfect}
Let $f$ be a symmetric Boolean function.
There exists a pair of distributions $\mu$ and $\nu$ that are $\left(K, K^{3/2}\cdot e^{-\Omega (\widetilde{\deg}_{1/3}(f)^2/K)} \right)$-indistinguishable for all $K \leq n/64$, but are $\Omega(1)$-reconstructible by $f$.
\end{corollary}

Corollary~\ref{cor:symmetricimperfect} is an immediate consequence of our Theorem~\ref{thm:mainupper}, and the
fact that any symmetric function has an optimal dual polynomial that is itself symmetric. 
In the special case $f = \AND$ (or equivalently $f = \OR$), Corollary~\ref{cor:symmetricimperfect} implies the existence of a \emph{visual secret sharing scheme} (see, for example~\cite{NaorS94}) that is $\left(K, K^{3/2}\cdot e^{-\Omega (n/K)} \right)$-statistically secure against all coalitions of size $K$, simultaneously for all $K$ up to size $n/64$. 
This property, where security guarantees are in place for many coalition sizes at the same time, is in contrast to an earlier result of Bogdanov and Williamson~\cite{BW17} where they proved that for any fixed coalition size $K$, there is a visual secret sharing scheme that is $(K, e^{-\Omega (n/K)})$-statistically secure. In their construction, the distribution of shares $\mu$ and $\nu$ depend on the value of $K$.

We remark that the bound of Corollary~\ref{cor:symmetricimperfect} cannot hold in general for $K = n$, since there exists distributions that are perfectly $\Omega(n)$-wise indistinguishable but are reconstructible by the majority function on all $n$ inputs.  We do not however know if a bound of the form $K \leq (1 - \Omega(1))n$ is tight in this context.

\paragraph*{Tight weight-degree tradeoffs for polynomials approximating symmetric functions.}
Let  $f \colon \{0, 1\}^n \to \{0, 1\}$ be any function. For any integer $d \geq 0$, denote by $W_\eps(f, d)$ the minimum \emph{weight} of any degree-$d$ polynomial that approximates $f$ pointwise to error $\eps$. By the weight of a polynomial, we mean the $\ell_1$-norm of its coefficients over the parity (Fourier) basis\footnote{In fact, our main weight lower bound (Corollary \ref{cor: wtdeglb}) holds over any set of functions (not just parities) that each depend on at most $d$ variables.}. 
In Section \ref{s:newdetails}, we observe that Corollary~\ref{cor:symmetricimperfect} implies weight-degree trade-off lower bounds for symmetric functions.

\begin{corollary}\label{cor: wtdeglb}
For any symmetric function $f \colon \{0, 1\}^n \to \{0, 1\}$, any constant $\eps \in (0, 1/2)$, and any integer $K$ such that $n/64 \geq K \geq \adeg_{\eps}(f)$, we have $W_\eps(f, K) \geq K^{-3/2} \cdot 2^{\Omega\left(\adeg_{1/3}(f)^2/K\right)}$.
\end{corollary}

The following theorem shows that the lower bound obtained in Corollary \ref{cor: wtdeglb} is tight (up to polylogarithmic factors in the exponent) for all symmetric functions.  

\begin{theorem}\label{thm: symub}
For any symmetric function  $f \colon \{0, 1\}^n \to \{0, 1\}$, any constant $\eps \in (0, 1/2)$ and $K > \adeg_\eps(f) \cdot \sqrt{\log n}$, $W_{\eps}(f, K) \leq 2^{\tilde{O}(\adeg_{1/3}(f)^2/K)}$.\footnote{Here and throughout, the $\tilde{O}$ notation hides polylogarithmic factors in $n$.}
\end{theorem}

Theorem \ref{thm: symub} also implies that Corollary \ref{cor:symmetricimperfect} is tight (up to polylogarithmic factors in the exponent)
for all symmetric $f$ and for all $K \geq \adeg_{1/3}(f) \sqrt{\log n}$. This is because any improvement to Corollary \ref{cor:symmetricimperfect}  would yield an improvement to Corollary \ref{cor: wtdeglb},
contradicting Theorem \ref{thm: symub}. 

\medskip
\noindent \textbf{Essentially Optimal Ramp Visual Secret Sharing Schemes.}
The following result shows that in the case $f=\AND$, 
Corollary \ref{cor:symmetricimperfect}  is essentially tight for \emph{all} $K \geq 2$, and Theorem~\ref{thm:mainupper} is tight as a reduction from perfect to approximate indistinguishability for symmetric distributions. It does so by constructing essentially optimal ramp visual secret sharing schemes.\footnote{A visual secret sharing scheme is a scheme where the reconstruction function is the $\AND$ of some subset of the shares. A ramp scheme is one where there is not necessarily a sharp threshold between the perfect secrecy and reconstruction thresholds; in particular, we allow for $K>k+1$.}  

\begin{theorem}
\label{thm:mainlowerprob}
For all $2\leq k < K\leq n$ there exist symmetric $k$-wise indistinguishable distributions $\mu$ and $\nu$ over $n$-bit strings that are $\sqrt{2^{-4K + 3} \cdot \sum_{d > k} \binom{2K}{K + d}^2}$-reconstructible by $\AND_K$, where $\AND_K(x)$ is the $\AND$ of the first $K$ bits of $x$.
\end{theorem} 

\noindent \emph{Discussion of Theorem \ref{thm:mainlowerprob}}.
This theorem gives the existence of a ramp visual secret sharing scheme that is perfectly secure against any $k$ parties, but in which any $K > k$ parties can reconstruct the secret with the above advantage. This generalizes the schemes in~\cite{BW17} where only reconstruction by all $n$ parties was considered. 

Let us express the reconstruction advantage appearing in Theorem \ref{thm:mainlowerprob}
in a manner more easily comparable to other results in this manuscript.
Standard results on anti-concentration of the Binomial distribution state that 
$2^{-2K} \cdot \sum_{d > k} \binom{2K}{K + d} = e^{-\Theta(k^2/K)}$ (see, e.g., \cite{kleinyoung}).
The Cauchy-Schwarz inequality then implies that the reconstruction advantage
appearing in Theorem \ref{thm:mainlowerprob} is at least  $K^{-1/2}\cdot e^{-O(k^2/K)}$.\footnote{
Theorem \ref{thm:mainlowerprob}
is closely related to Theorem \ref{thm:dualpoly},
in that Theorem \ref{thm:mainlowerprob} gives \emph{another} anti-concentration-based proof that $\adeg_{\eps}(\AND_K) \geq k$ for $\eps=K^{-1/2}\cdot e^{-\Theta(k^2/K)}$.
However, the two results are incomparable.
Theorem \ref{thm:mainlowerprob} does not yield an explicit dual polynomial for $\AND_K$, and 
the $\eps$-approximate degree lower bound for $\AND_K$ implied by Theorem \ref{thm:mainlowerprob} 
is loose by the $K^{-1/2}$ factor appearing in the expression for $\eps$.
On the other hand, Theorem \ref{thm:dualpoly} only yields a visual secret sharing scheme
with reconstruction by all $n$ parties,  while Theorem \ref{thm:mainlowerprob} 
yields a ramp scheme with non-trivial reconstruction advantage by the $\AND$ of the first $K$ (out of $n$) parties.
}

Hence, the visual secret sharing schemes given in Theorem~\ref{thm:mainlowerprob} are nearly optimal; if the reconstruction advantage could be improved by more than the leading $\poly(K)$ factor (or the constant factor in the exponent), then this would contradict Theorem~\ref{thm:mainupper} which upper bounds the distinguishing advantage of any statistical test over $K$ bits against symmetric, perfectly $k$-wise indistinguishable distributions. Theorem~\ref{thm:mainlowerprob} also shows that the indistinguishability parameter in Theorem~\ref{thm:mainupper} cannot be significantly improved, even in the restricted case where the only statistical test is $\AND_K$. 

\medskip
In Section~\ref{sec:robustness} we describe another application of Theorem~\ref{thm:mainupper} to security against share consolidation and ``downward self-reducibility'' of visual secret shares.

\subsection{Related Works} 

\medskip \noindent \textbf{Prior Work.} Servedio, Tan, and Thaler \cite{STT} established Corollary \ref{cor: wtdeglb} and Theorem \ref{thm: symub}
in the special case $f=\OR$, showing that degree $d$ polynomials that approximate the $\OR$ function
require weight $2^{\tilde{\Theta}(n/d)} = 2^{\tilde{\Theta}(\adeg_{1/3}(\OR)^2/d)}$.\footnote{These bounds
for $\OR$ were implicit in \cite{STT}, but not explicitly highlighted. The upper bound was explicitly stated in \cite[Lemma 4.1]{chandrasekaran2014faster},
which gave applications to differential privacy, and the lower
bound in \cite[Lemma 32]{BT13eccc}, which used it to establish tight weight-degree tradeoffs for polynomial threshold functions computing read-once DNFs.}
 They used
this result to establish tight weight-degree tradeoffs for polynomial threshold functions computing decision lists.
As previously mentioned, Bogdanov and Willamson \cite{BW17} generalized the weight-vs-degree lower bound from \cite{STT} beyond polynomials,
thereby
obtaining a visual secret-sharing scheme  for any fixed $K$ that is $(K, e^{-\Omega (n/K)})$-statistically secure.

Elkies \cite{elkiesmathoverflow} and Sachdeva and Vishnoi \cite{sachdeva2013approximation} exploit concentration of measure to prove a tight upper bound
on the degree of univariate polynomials that approximate the function $t \mapsto t^n$ over the domain $[-1, 1]$.
Their techniques inspired our (much more technical) proof of Theorem \ref{thm:mainupper}.

 \medskip \noindent \textbf{Other Related Work.} This work subsumes Bogdanov's manuscript~\cite{Bog18}, which shows a slightly weaker lower bound on the weighted approximate degree of AND, and does not derive an explicit dual polynomial.  In independent work, Huang and Viola \cite{hv19} prove a weaker form of our Corollary~\ref{cor:symmetricimperfect}:  their distributions $\mu, \nu$ depend on the value of $K$.  They also prove (a slightly tighter version of) Theorem \ref{thm: symub}, thereby establishing that the statistical distance in Corollary~\ref{cor:symmetricimperfect} is tight.

\subsection{Techniques and Organization} 

The proof of Theorem~\ref{thm:dualpoly} (Section~\ref{s:weighteddeg}) is an elementary verification that the function $\phi$ given in \eqref{eq:dualq} is a dual polynomial.  The only property that is not immediate is correlation with $\AND$.  Verifying this property amounts to upper bounding the normalization constant $Z$, which follows from orthogonality of the Fourier characters.

In the proof of Theorem~\ref{thm:mainupper} (Section~\ref{s:apxfromperfect}), a $K$-bit statistical distinguisher for symmetric distribution is first decomposed into a sum of at most $K+1$ tests $Q_w$ that evaluate to 1 only when the input has Hamming weight exactly $w$.  Lemma~\ref{lemma:main3} shows that the univariate symmetrizations $p_w$ of these distinguishers can be pointwise approximated by a degree-$k$ polynomial with error at most $O(K^{1/2}) \cdot e^{-\Omega(k^2/K)}$.  

To construct the desired approximation, we derive an identity relating the moment generating function of the squared Chebyshev coefficients of $p_w$ (interpreted as relative probabilities) to the average magnitude of a polynomial $g$ related to $p_w$ on the unit complex circle (Claims~\ref{claim:regcoeffs} and~\ref{claim:normg}).  We bound these magnitudes analytically (Claim~\ref{claim:fupper}) and derive tail inequalities for the Chebyshev coefficients from bounds on the moment generating function as in standard proofs of Chernoff-Hoeffding bounds.

In the special case when the secrecy parameters $k$ and $K$ are fixed and the number of parties $n$ approaches infinity, $p_w(t)$ turns out to equal  $C_w (t - 1)^w (t + 1)^{K - w}$, where $C_w$ is some quantity independent of $t$. In this case, the Chebyshev coefficients are the regular coefficients of the polynomial $g^\infty(s) = 2^{-w} C_w (s - 1)^{2w}(s + 1)^{2(K-w)}$.\footnote{The $i$-th coefficient of $g^\infty$ is the value of the $i$-th Kravchuk polynomial with parameter $2K$ evaluated at $2w$.}  When $w = 0$, $K/2$, or $1$, the coefficients of $g^\infty$ are exponentially concentrated around the middle as they follow the binomial distribution.  We prove that this exponential decay in magnitudes happens for all values of $w$, which requires understanding complicated cancellations in the algebraic expansion of $g^\infty(s)$.
We generalize this analysis to the finitary setting $n \geq 64K$.

We prove Theorem \ref{thm: symub} (Section \ref{s:symub}) by writing any symmetric function $f$ as a sum of at most $\ell := \min\{|f^{-1}(0)|, |f^{-1}(1)|\}$ many conjunctions, and approximating each conjunction to such low error (namely error $\ll \ell$) that the sum of all approximations is an approximation for $f$ itself. Theorem \ref{thm: symub} then follows by constructing
 low-weight, low-degree polynomial
approximations for each conjunction in the sum. 

Theorem~\ref{thm:mainlowerprob} (Section~\ref{s:mainlowerprob}) is proved by lower bounding the error of degree $k$ polynomial approximations to the symmetrization $f$ of the function $\AND_K\left(x|_{\{1, \dots, K\}}\right)$.  By duality, a lower bound on approximation error translates into a secret sharing scheme with the same reconstruction advantage.  To lower bound the error, we estimate the values of the coefficients in the Chebyshev expansion of $f$ with indices larger than $k$.  Owing to orthogonality, the largest of these coefficients lower bounds the approximation error of any degree-$k$ polynomial.

In Section~\ref{sec:robustness} we formulate a security of secret sharing against consolidation and downward self-reducibility of visual schemes, and derive these properties from the main results. 

\section{Dual Polynomial For the Weighted Approximate Degree of AND}
\label{s:weighteddeg}

In this section we prove Theorem~\ref{thm:dualpoly} and derive its two corollaries about the unweighted and weighted approximate degree of AND.  

\medskip \noindent \textbf{Notation and Definitions.} 
Let $[n]=\{1, \dots, n\}$. Given a vector $w \in \R_{\geq 0}^n$, define the weight of a monomial $\chi_{S}(x) = \prod_{i \in S} x_i, x_i \in \{-1, 1\}$ to equal $\sum_{i \in S}w_i$.  Define the $w$-weighted degree of a polynomial to be the maximum weight of a monomial in it.  That is, if $p = \sum_{S \subseteq [n]}c_S\chi_S$, then define
\[
\deg_w(p) = \max_{S : c_S \neq 0}w(S).
\]
Define the $w$-weighted $\eps$-approximate degree $\adeg_{w, \eps}(f)$ to be the minimum $w$-weighted degree of a polynomial $p$ that satisfies $\abs{p(x) - f(x)} \leq \eps$ for all $x$ in the domain of $f$.
Given two real-valued functions $f, g$ over domain $\{-1, 1\}^n$, define $\langle f, g \rangle := \frac{1}{2^n}\sum_{x \in \{-1, 1\}^n} f(x) \cdot g(x)$. 

\begin{lemma}\label{lem:dualpoly}
For any finite set $X$ and any function $f \colon X \rightarrow \mathbb{R},~\adeg_{w, \eps}(f) \geq d$ iff there exists a function $\phi: X \rightarrow \mathbb{R}$ satisfying the following conditions.
\begin{itemize}
   \item \emph{Pure high degree}: For any real polynomial $p$ of weighted degree is at most $d$,~$\langle \phi, p \rangle = 0$.
    \item \emph{Normalization}: $\sum_{x \in X}|\phi(x)| = 1$,
    \item \emph{Correlation}: $\langle \phi, f\rangle  \geq \eps$,
\end{itemize}
\end{lemma}

We call $\phi$ a dual witness for $\adeg_{w, \eps}(f) \geq d$.  The lemma follows by linear programming duality and is a straightforward generalization of previous results (see e.g.~\cite{Spalek08, BT13}). We prove the ``if'' direction, which is sufficient for our purposes.

\begin{proof}
For any $p$ of weighted degree at most $d$,
\[ \|f - p\|_{\infty} = \|f - p\|_{\infty} \|\phi\|_1 \geq \ip{\phi, f - p} = \ip{\phi, f} - \ip{\phi, p} \geq \eps. \]
\end{proof}

The dual polynomial of interest is 
\[
\phi(x) = \frac{(-1)^n}{Z}\chi_{[n]}(x) \cdot \E_{S \sim \H}[\chi_S(x)]^2,
\]
where $x \in \{-1, 1\}^n$, $\H$ is the uniform distribution over the sets $\bra{S \subseteq [n] : w(S) \leq (\|w\|_1 - d)/2}$, and $Z$ is the normalization constant 
\[ Z = \sum_{x \in \pmone^n} \E_{S \sim \H}[\chi_S(x)]^2. \]

\begin{proof}[Proof of Theorem~\ref{thm:dualpoly}]
We prove the theorem by showing that $\phi$ satisfies the three conditions of Lemma~\ref{lem:dualpoly}.  The expression $\E_{S \sim \H}[\chi_S(x)]^2$ can be written as a sum of products of pairs of monomials of weight at most $(\|w\|_1 - d)/2$, so its weighted degree is at most $\|w\|_1 - d$.  Thus every monomial that occurs in the expansion of $\chi_{[n]}(x) \E_{S \sim \H}[\chi_S(x)]^2$ must have weighted degree {\em at least} $d$, and so $\phi$ has pure high weighted degree at least $d$ as desired.

The scaling by $Z$ in the definition of $\phi$ ensures that $\phi$ has $L_1$ norm 1.  The correlation of $\phi$ and $\AND$ is given by
$ \ip{\phi, \AND} = \phi(1^n) = \frac{1}{Z}. $
Finally, the normalization constant $Z$ evaluates to
\begin{align*} 
Z &= \sum_{x \in \{-1, 1\}^n} \E_{S \sim \H}[\chi_S(x)]^2
  = \sum_{x \in \{-1, 1\}^n} \E_{S \sim \H}[\chi_S(x)]\E_{T \sim \H}[\chi_T(x)] \\ 
  &= \sum_{x \in \{-1, 1\}^n} \E_{S, T \sim \H}[\chi_{S \Delta T}(x)] 
  = \E_{S, T \sim \H} \sum_{x \in \{-1, 1\}^n} \chi_{S \Delta T}(x)  \\
  &= 2^n \pr[S = T] 
  = \frac{2^n}{\abs{\H}},
\end{align*}
since the inner summation over $x$ evaluates to $2^n$ when $S = T$, and zero otherwise.

It remains to show that $1/Z = \abs{\H}/2^n$ equals the desired expression for $\eps$.  For a set $S \subseteq [n]$, let $X(S) \in \{-1, 1\}^n$ be the string that assigns values $1$ and $-1$ to elements inside and outside $S$, respectively.  Then $w(S) = \|w\|_1/2 + \ip{w, X(S)}/2$, so
\[
\frac{\abs{\H}}{2^n} 
   = \pr_{S \subseteq [n]}[w(S) \geq \|w\|_1/2 + d/2] 
   = \pr_{X \sim \{-1, 1\}^n}[\ip{w, X} \geq d]. 
\]
\end{proof}

\begin{corollary}[Approximate degree of AND]
\label{cor:and}
Recall that $\AND \colon \{-1, 1\}^n \to \{0, 1\}$ denotes the function satisfying $\AND(x)=1$ if and only if $x=1^n$. 
If $p$ has degree at most $d$, then $|p(x) - \AND(x)| \geq \pr[X \leq (n - d)/2]$ for some $x$, where $X$ is a $\mathrm{Binomial}(n, 1/2)$ random variable.
\end{corollary}

The expression on the right is lower bounded by the larger of $1/2 - O(d/\sqrt{n})$ and $2^{-O(d^2/n)}$.  In the large $d$ regime ($d \geq \sqrt{n}$), this bound is tight ~\cite{KahnLS96, BCdWZ99}

\begin{proof}
Apply Theorem~\ref{thm:dualpoly} to the weight vector $w = (1, 1, \dots, 1)$.  
\end{proof}

Earlier constructions of dual polynomials for AND are quite different from our Corollary~\ref{cor:and}  \cite{KahnLS96, Spalek08, BT13, sherstov15} and are based on real-valued polynomial interpolation. Specifically, for a carefully chosen set $T\subseteq \{0, 1, \dots, n\}$ of size $|T|=2d$, the prior constructions consider a \emph{univariate} polynomial
$p(t) = \prod_{i \in [n] \setminus T} (t-i)$,
and they define $\psi(x) = p(|x|),$ where $|x|$ denotes the Hamming weight of $x$. Clearly $\psi$ has degree at most $n-|T|$. A fairly complicated
calculation is required to show that, for an appropriate choice of $T$, defining $\psi$ in this way
ensures that $|\psi(1^n)|$ captures an $\varepsilon$-fraction of the $L_1$-mass of $\psi$.

\label{s:discusss}

\begin{corollary}[Weighted approximate degree of AND]
\label{cor:weighted}
$\widetilde{deg}_{w, 3/32}(\AND) \geq \|w\|_2/2$.
\end{corollary}

The proof uses the Paley-Zygmund inequality:

\begin{lemma}[Paley-Zygmund inequality]\label{lem: pz}
Let $Z \geq 0$ be any random variable with finite variance.  Then, for any $0 < \theta < 1$,
\[
\Pr[Z \geq \theta\E(Z)] \geq (1 - \theta)^2\frac{(\E[Z])^2}{\E[Z^2]}.
\]
\end{lemma}

\begin{proof}[Proof of Corollary~\ref{cor:weighted}]
We apply the Paley-Zygmund inequality to $\ip{w, X}^2$.  First, $\E[\ip{w, X}]^2 = \|w\|_2^2$ and $\E[\ip{w, X}^4] = \sum w_i^4 + 3 \sum w_i^2w_j^2 \leq 3\|w\|_2^2$.  Then
\[
\Pr\left[\ip{w, X} \geq \frac{\|w\|_2}{2}\right] = 
\frac12 \Pr\left[\abs{\ip{w, X}} \geq \frac{\|w\|_2}{2}\right] = 
\frac12 \Pr\left[\ip{w, X}^2 \geq \frac{\|w\|_2^2}{4}\right] \geq \frac12 \cdot \frac{9}{16}\cdot\frac{1}{3} = \frac{3}{32},
\]
where the first equality follows from the sign-symmetry of $X$.  Applying Theorem~\ref{thm:dualpoly} with $d = \|w\|_2/2$ yields the claim.
\end{proof}

\section{Approximate Indistinguishability from Perfect Indistinguishability}
\label{s:apxfromperfect}

In this section, we prove Theorem \ref{thm:mainupper}, which states that any pair of symmetric and perfectly $k$-wise indistinguishable distributions over $\{0, 1\}^n$ are also approximately indistinguishable against statistical tests that observe $K > k$ of the bits.  We may and will assume without loss of generality that the statistical test is a symmetric function, meaning that it depends only on the Hamming weight of the observed bits of its input.

Let $X$ and $Y$ denote an arbitrary pair of symmetric $(k, 0)$-wise indistinguishable distributions over $\{0, 1\}^n$. We will be interested in obtaining an upper bound on the statistical distance of their projections to any $K$ indices of $[n]$, namely the advantage $\E_X[T(X|_S)-\E_Y[T(Y|_S)]$ where $T:\{0,1\}^K \rightarrow \{0,1\}$ is a symmetric function and $S \subseteq [n]$ is any set of size $K$.  We can decompose $T$ into a sum of tests $Q_w:\{0,1\}^K \rightarrow \{0,1\}$, where $Q_w$ 
outputs 1 if and only if the Hamming weight of its input is exactly $w$.  Specifically, we decompose $T$ as 
\begin{equation}
\label{eq:expT}
 T = \sum_{w=0}^K b_wQ_w,
\end{equation}
where each $b_w$ is either zero or one.  We will bound the distinguishing advantage of each $Q_w$ in the sum individually.  This advantage is captured by a univariate function $p_w$ that expresses $Q_w$ in terms of the Hamming weight of its input, after shifting and scaling the Hamming weight to reside in the interval $[-1, 1]$.

\begin{fact}
\label{fact:pw}
Let $S\subseteq [n]$ be any set of size $K$. There exists a univariate polynomial $p_w$ of degree at most $K$ such that
the following holds.
For all $t \in \{-1, -1+2/n, \dots, 1-2/n, 1\}$, $p_w(t) = \E_{Z}[Q_w(Z|_S)]$ where $Z$ is a random string of Hamming weight $\phi^{-1}(t) = (1 - t)n/2 \in \{0, 1, \dots, n\}$.
\end{fact}
\begin{proof}
This statement is a simple extension of Minsky and Papert's classic symmetrization technique \cite{MiP69}.
Specifically, Minsky and Papert showed that for any polynomial $p_n \colon \{0, 1\}^n \to \R$,
there exists a univariate polynomial $P$ of degree at most the total degree of $p_n$, such
that for all $i \in \{0, \dots, n\}$, $P(i) = \mathbb{E}_{|x|=i}[p_n(x)]$. 
Apply this result to $p_n(x) = Q_w(x|_S)$ and let $p_w(t) = P(\phi^{-1}(t)) = P\left((1 - t)n/2\right)$.
The fact then follows from the observation that the total degree of $Q_w(x|_S)$ is at most $K$,
since this function is a $K$-junta. 
\end{proof}

In particular, the value $p_w(t)$ is a probability for every $t \in \{-1, -1+2/n, \dots, 1-2/n, 1\}$.
Moreover, this probability must equal zero when the Hamming weight of $Z$ is less than $w$ or greater than $n - K + w$.  Therefore $p_w$ has $K$ distinct zeros at the points $Z_w = Z_{-} \cup Z_{+}$, where
\begin{equation}
\label{eq:zeros}
Z_{-}=\left\{-1+ 2h/n: h=0,...,K-w-1\right\}, \qquad Z_{+} = \{1 - 2h/n: h=0,...,w-1\}.
\end{equation}
and so $p_w$ must have the form
\begin{equation}
\label{eq:pw}
p_w(t) = C_w \cdot \prod_{z \in Z_w} (t - z)
\end{equation}
for some $C_w$ that does not depend on $t$.\footnote{$p_w$, $C_w$, and $Z_w$ also depend on $K$ and $n$ but we omit those arguments from the notation as they will be fixed in the proof.}  As $p_w(t)$ is probability when $t \in \{-1, -1 + 2/n, \dots, 1 - 2/n, 1\}$, the function $p_w$ is 1-bounded at those inputs.  In fact, $p_w$ is uniformly bounded on the interval $[-1, 1]$:

\begin{claim}
\label{claim:bounded}
Assuming $n \geq 64K$, $\abs{p_w(t)} \leq 2$ for all $t \in [-1, 1]$.
\end{claim}

The proof is in Section~\ref{sec:fupper}.  Formula \eqref{eq:pw} and Claim~\ref{claim:bounded} will be applied to show that $p_w$ has a good uniform polynomial approximation on the interval $[-1, 1]$.  

\begin{lemma}
\label{lemma:main3}
Assuming $n \geq 64K$, there exists a degree-$k$ polynomial $q_w$ such that $\abs{p_w(t) - q_w(t)} \leq 4\sqrt{K} \exp(-k^2/1156K)$ for all $t \in [-1, 1]$.
\end{lemma}

Lemma~\ref{lemma:main3} is the main technical result of this section.  It is proved in Section~\ref{sec:main3outline}. 

\begin{proof}[Proof of Theorem~\ref{thm:mainupper}]
Now let $T$ be a general distinguisher on $K$ inputs.  By Facts~\ref{fact:symmetricmarginals} and~\ref{fact:symmetrictest} (see Appendix), $T$ can be assumed to be a symmetric Boolean-valued function.  We bound the distinguishing advantage as follows. Recalling that $X$ and $Y$ are $(k, 0)$-indistinguishable symmetric distributions over $\{0, 1\}^n$, for any set $S \subseteq [n]$ of size $K$ we have:
\begin{multline*}
\E[T(X|_S)] - \E[T(Y|_S)]  \\
\begin{aligned} 
  &= \sum_{w = 0}^K b_w \bigl(\E[Q_w(X|_S)] - \E[Q_w(Y|_S)]\bigr)  \qquad\text{(by~\eqref{eq:expT})} \\
  &\leq \sum_{w=0}^K \bigabs{\E[Q_w(X|_S)] - \E[Q_w(Y|_S)]}   \qquad\text{(by boundedness of $b_w$)} \\
  &= \sum_{w=0}^K \bigabs{\E[p_w(\phi(\abs{X})] - \E[p_w(\phi(\abs{Y}))]} \qquad\text{(by symmetry of $X, Y$, and Fact~\ref{fact:pw})} \\
  &\leq \sum_{w=0}^K \bigabs{\E[q_w(\phi(\abs{X}))] - \E[q_w(\phi(\abs{Y}))]} + 8\sqrt{K} \exp(-k^2/1156K)
  \qquad\text{(by Lemma~\ref{lemma:main3})} \\
  &= O(K^{3/2}) \cdot e^{-k^2/1156K} \qquad\text{(by $k$-wise indistinguishability of $X, Y$)}
\end{aligned}
\end{multline*}

Therefore, $X$ and $Y$ are $(K, O(K^{3/2})\cdot e^{-k^2/1156K})$-wise indistinguishable for $2 \leq K\leq n/64$.
\end{proof}

\subsection{Proof of Lemma~\ref{lemma:main3}}
\label{sec:main3outline}

We will prove Lemma~\ref{lemma:main3} by studying the Chebyshev expansion of $p_w$.  To this end we take a brief detour into Chebyshev polynomials and an even briefer one into Fourier analysis.

\paragraph*{Chebyshev polynomials.}
The Chebyshev polynomials are a family of real polynomials $\{T_d\}$, 1-bounded on $[-1, 1]$, with $T_d$ having degree $d$.  We extend the definition to negative indices by setting $T_{-d} = T_d$.  The Chebyshev polynomials are orthogonal with respect to the measure $d\sigma(t) = (1 - t^2)^{-1/2}dt$ supported on $[-1, 1]$.  Therefore every degree-$K$ polynomial $p\colon \R \to \R$ has a unique (symmetrized) Chebyshev expansion
\[ p(t) = \sum_{d = -K}^K c_d T_d(t), \qquad c_{-d} = c_d \] 
where $c_{-K}, \dots, c_K$ are the {\em Chebyshev coefficients} of $p$.  

The Chebyshev polynomials satisfy the following identity, which plays an important role in our analysis:

\begin{fact}
\label{fact:chebid}
$t \cdot T_d(t) = \frac12 T_{d-1}(t) + \frac12 T_{d+1}(t)$.
\end{fact}
This formula, together with the ``base cases'' $T_0(t) = 1$ and $T_1(t) = t$, specifies all Chebyshev polynomials.

\medskip\noindent
We will also need the following form of Parseval's identity for univariate polynomials.

\begin{claim}[Parseval's identity]
\label{claim:parseval}
For every complex polynomial $h$, the sum of the squares of the magnitudes of the coefficients of $h$ equals $\E_z[\abs{h(z)}^2]$, where $z$ is a random complex number of magnitude 1.
\end{claim}

\paragraph*{Proof outline.}
We will argue that the Chebyshev expansion $\sum_{d=-K}^Kc_dT_d(t)$ of $p_w(t)$ has small weight on the coefficients $c_d$ when $\abs{d} > k$.  Zeroing out those coefficients then yields a good degree-$k$ approximation of $p_w$ as desired.

The upper bound on the Chebyshev coefficients of $p_w$ is derived in two steps.  The first step, which is of an algebraic nature, expresses the Chebyshev coefficients of $p_w$ as regular coefficients of a related polynomial $g$.\footnote{We omit the dependence on $w$ as this parameter remains constant throughout the proof.}  We are interested in the coefficients of the derived polynomial $g_\eps(s) = g((1 + \eps)s)$, which represent the Chebyshev coefficients $c_d$ of $p_w$ amplified by the exponential scaling factor $(1 + \eps)^d$.

The second step, which is analytic, upper bounds the magnitude of the coefficients of $g_\eps(s)$.  The main tool is Parseval's identity, which identifies the sum of the squares of these coefficients by the average magnitude of $g_\eps$ over the complex unit circle $\E_\theta |g((1 + \eps)e^{i\theta})|^2$.  We bound the {\em maximum} magnitude $\max_\theta |g((1 + \eps)e^{i\theta})|^2$ by explicitly analyzing the function $g$.  This step comprises the bulk of our proof.

The third step translates the bound on the squared 2-norm $\sum_{d = -K}^K (1 + \eps)^{2d} c_d^2$ of the amplified coefficients into a tail bound on $c_d$ by optimizing over a suitable value of $\eps$.  This is analogous to the standard derivation of Chernoff-Hoeffding bounds by analysis of the moment generating function of the relevant random variable.

We now sketch how this outline is executed for the special case where $n$ tends to infinity while $k$ and $K$ remain fixed. Although this setting is technically much easier, it allows us to highlight the main conceptual points of our argument.  The analysis for finite $n$ can be viewed as an approximation of this proof strategy.

\paragraph*{Sketch of the limiting case $n\rightarrow\infty$.}
By the expansion~\eqref{eq:pw} of $p_w$, as $n$ tends to infinity $p_w$ converges uniformly to the function
\[
p_w^{\infty}(t) = C_w \cdot (t-1)^w(t+1)^{K-w},
\]
as this corresponds to Fact \ref{fact:pw} when the bits of the string $Z$ are independent and $(1-t)/2$-biased.  As $p^{\infty}_w(t)$ is a probability for every $t \in [-1, 1]$, Claim~\ref{claim:bounded} follows immediately.

\medskip\noindent\textbf{Step 1.} Our algebraic treatment of the Chebyshev transform yields that the Chebyshev coefficient $c_d$ of $p_w^{\infty}$ is the $(K+d)$-th regular coefficient of the polynomial 
\begin{equation}
\label{eq:ginfty}
g^{\infty}(s) = C_w\left(\frac{s-1}{\sqrt{2}}\right)^{2w}\left(\frac{s+1}{\sqrt{2}}\right)^{2(K-w)}.
\end{equation}

\medskip\noindent\textbf{Step 2.} The evaluation of the polynomial $g_\eps^\infty(s) = g^\infty((1 + \eps)s)$ at $s = e^{i\theta}$ satisfies the identity 
\begin{equation}
\label{eq:gcircle}
\left|g^\infty\left((1+\varepsilon)e^{i\theta}\right)\right|=(1+\varepsilon)^K\cdot (1+\delta)^K\cdot C_w\cdot \left(1-\frac{\cos\theta}{1+\delta}\right)^w\left(1+\frac{\cos\theta}{1+\delta}\right)^{K-w},
\end{equation}
where $\delta = \eps^2 / 2(1 + \eps)$.  This happens to equal 
\begin{equation}
\label{eq:selfsim}
(1+\varepsilon)^K (1+\delta)^K p_w(\cos \theta / (1 + \delta)),
\end{equation}
and is in particular uniformly bounded by $(1+\varepsilon)^K (1+\delta)^K$ for all $\theta$.  This similarity between $p^\infty$ and $g_\eps^\infty$ is the crux of our analysis.

\medskip\noindent\textbf{Step 3.} By Parseval's identity, after suitable shifting and cancellation, the amplified sum of Chebyshev coefficients $\sum_{d = -K}^K (1 + \eps)^{2d}c_d^2$ is upper bounded by $(1 + \delta)^{2K}$.  Therefore the tail $\sum_{k \geq d} c_d^2$ can have value at most $(1 + \delta)^{2K}/(1 + \eps)^{2k} \leq \exp (2K \eps^2 - 2(\eps - \eps^2/2)k)$.  This upper bound holds for all $\eps \in [0, 1]$, and plugging in the approximate minimizer $\eps = k/2K$ yields a bound of the desired form $\exp(-\Omega(k^2/K))$.

\paragraph*{Outline of the general case.}
We now give the outline of our full proof for the general case and relevant technical statements that we use to prove our main upper bound.  Identity~\eqref{eq:ginfty} generalizes to the following statement:

\begin{claim}
\label{claim:regcoeffs}
The Chebyshev coefficient $c_d$ of $p_w$ is the $(K+d)$-th regular coefficient of the polynomial
\[g(s) = C_w \prod_{z\in Z_w} \biggl(\frac{s^2 - 2sz+1}{2}\biggr),\]
{where $C_w$ is as in Equation \eqref{eq:pw}.} 
\end{claim}

The general form of identity~\eqref{eq:gcircle} is:

\begin{claim}
\label{claim:normg}
For $\eps > 0$, $\delta = \eps^2/2(1 + \eps)$, and $\theta \in [-\pi, \pi]$, 
\[ \big|g((1 + \eps) e^{i\theta})\big|^2 = (1+\varepsilon)^{2K} (1+\delta)^{2K} \cdot C_w^2 \prod_{z \in Z_w} h_{\delta(1 + 1/(1 + \delta))}\biggl(\frac{\cos \theta}{1 + \delta}, z\biggr) \]
where $h_{\delta}(s, z) = (s - z)^2 + \delta (1 - z^2)$.
\end{claim}

Owing to the second term in $h_\delta$, there is no identity analogous to~\eqref{eq:selfsim} when $n$ is finite and $p_w$ has zeros inside $(-1, 1)$.  Nevertheless, $\prod_{z\in Z_w} h_{\delta}(s, z)$ can be uniformly bounded either by a sufficiently small multiple of $p_w(s)^2$, or a fixed quantity that is constant in the parameter range of interest.

\begin{claim}
\label{claim:fupper} Assume $n \geq 64K$ and $w \leq K/2$.  Then 
\[ C_w^2\cdot \prod_{z\in Z_w} h_\delta(s, z) \leq\begin{cases} 
      e^{65\delta K} \cdot p_w(s)^2  &\text{if $\abs{s} \leq 1 - w/16K$}\\
      e^{65\delta K} &\text{if $1 - w/16K \leq \abs{s} \leq 1$.}
\end{cases}\]
\end{claim}

We now prove Lemma~\ref{lemma:main3}.  Claim~\ref{claim:regcoeffs} is proved in Section~\ref{sec:regcoeffs}.  Claim~\ref{claim:normg} is proved in Section~\ref{sec:normg}.  Claims~\ref{claim:bounded} and~\ref{claim:fupper} are proved in Section~\ref{sec:fupper} as the proofs share the same structure.  

\begin{fact}
\label{fact:symm}
$p_w(t) = p_{K - w}(1 - t)$.
\end{fact}
\begin{proof}
By Fact~\ref{fact:pw}, both sides are degree-$K$ polynomials that agree on $n + 1 > K$ points so they are identical.
\end{proof}

\begin{proof}[Proof of Lemma~\ref{lemma:main3}]  By Fact~\ref{fact:symm} we may and will assume that $w \leq K/2$.  Let $p_w = \sum_{d = -K}^K c_dT_d$.  The approximating polynomial $q_w$ is $\sum_{|d| < k} c_d T_d$.  It remains to prove a tail upper bound on the Chebyshev coefficients. By Claim~\ref{claim:regcoeffs}, the $(K + d)$-th coefficient of $g(s)$ is $c_d$.  Therefore the polynomial $g_\eps(s) = g((1 + \eps)s)$ has coefficients $(1 + \eps)^{K + d}c_d$ as $d$ ranges from $-K$ to $K$.  We apply Parseval's identity (Claim~\ref{claim:parseval}) to $g_\eps$.

It follows that
\begin{align*}
\sum_{d = -K}^K (1 + \eps)^{2(K + d)} c_d^2 
  &= \E_\theta |g((1 + \eps)e^{i\theta})|^2   \\
  &\leq \max_{\theta \in [-\pi, \pi]} |g((1 + \eps)e^{i\theta})|^2 \\
  &= \max_{s \in [-1, 1]} (1+\varepsilon)^{2K} (1+\delta)^{2K} \cdot C_w^2 \prod_{z \in Z_w} h_{\delta(1 + 1/(1 + \delta))}(s/(1 + \delta), z),
\end{align*}
by Claim~\ref{claim:normg}.  Since $0 \leq \delta = \varepsilon^2/2(1+\varepsilon) \leq 1/2$, for simplicity we may replace $h_{\delta(1 + 1/(1 + \delta))}(s/(1 + \delta), z)$ by $h_{2\delta}(s, z)$ in the above inequality.  This gives the following approximation bound for $\alpha = \max_{t \in [-1, 1]} |p_w(t) - q_w(t)|$:

\begin{align*} 
\alpha &= \max_{t \in [-1, 1]} \Big|\sum\nolimits_{\abs{d} \geq k} c_d T_d(t)\Big| \\
  &\leq \sum\nolimits_{\abs{d} \geq k} \abs{c_d} \max_{t \in [-1, 1]} \abs{T_d(t)} \\
  &\leq 2\sum_{d \geq k} \abs{c_d} \qquad\text{(by symmetry and boundedness of $T_d$)} \\
  &\leq 2\sqrt{K} \cdot \sqrt{\sum\nolimits_{d \geq k} c_d^2} \qquad\text{(by Cauchy-Schwarz)} \\
  &\leq 2\sqrt{K} \cdot \sqrt{(1 + \eps)^{-2(K + k)} \sum\nolimits_{d \geq k} (1 + \eps)^{2(K + d)}c_d^2} \\
  &\leq 2\sqrt{K} \sqrt{(1 + \eps)^{-2k} \cdot  (1 + \delta)^{2K} \cdot \max_{s \in [-1, 1]} C_w^2 \prod_{z \in Z_w} h_{2\delta}(s, z)}.
\end{align*}

By the boundedness of $p_w$ (Claim~\ref{claim:bounded}), the upper bounds in Claim~\ref{claim:fupper} can be unified by the inequality
\[ C_w^2 \prod_{z \in Z_w} h_{2\delta}(s, z) \leq 4 e^{130 \delta K} \]
that is valid for all $s \in [-1, 1]$.  Since $1 + \delta \leq e^{\delta}$ and $1 + \eps \geq e^{\eps - \eps^2/2}$ for $0 \leq \eps \leq 1$,
\[
\alpha 
  \leq 2\sqrt{K} \cdot \sqrt{\frac{(1+\delta)^{2K}}{(1+\varepsilon)^{2k}}\cdot 4 e^{130 \delta K}}
  \leq 4\sqrt{K} \cdot \sqrt{e^{132 \delta K - 2 \varepsilon k + \eps^2k}} \leq 4\sqrt{K} \cdot \sqrt{e^{67\eps^2 K - 2\eps k}},
\]
where the last inequality follows from the definition $\delta = \varepsilon^2/2(1+\varepsilon)$.  Setting $\eps = k/34K$ we obtain that $\alpha \leq 4\sqrt{K} \cdot e^{-k^2/1156K}$.
\end{proof}

\subsection{Proof of Claim~\ref{claim:regcoeffs}}
\label{sec:regcoeffs}

Claim~\ref{claim:regcoeffs} is a direct consequence of the following formula for the Chebyshev expansion of products of linear functions.

\begin{claim} \label{convenientform}
If $p(t) = \prod_{z \in Z} (t - z)$, where $\abs{Z} = K$ then the $d$-th Chebyshev coefficient of $p$ is the $d$-th regular coefficient of the Laurent polynomial $g(s) = \prod_{z \in Z} (s + s^{-1} - 2z)/2$.
\end{claim}

Indeed, multiplying the polynomial $g(s)$ from Claim \ref{convenientform}  by $s^K$ then yields Claim~\ref{claim:regcoeffs}.

\begin{proof}
We prove this by induction on $K$.  When $K = 0$, $p$ has only one nonzero Chebyshev coefficient and it is equal to $1$ as claimed.  Now assume the claim holds for $p(t)$ and we prove it for $(t - z)p(t)$.  Let $[s^d] \left(g(s)\right)$ denote the $d$-th regular coefficient of $g$.  Then the Chebyshev expansion of $p$ is
\[ p(t) = \sum_d  [s^d]\left(g(s)\right) \cdot T_d(t), \] 
and the Chebyshev expansion of $(t - z)p(t)$ is
\begin{align*}
\!\!\!\!\!\!\!\!\!\!\!\!\!(t - z)p(t) &= \sum_d [s^d]\left(g(s)\right) t T_d(t) - \sum_d [s^d]\left(g(s)\right)  z T_d(t) \\
  &= \sum_d [s^d]\left(g(s)\right) \cdot \tfrac12 T_{d-1}(t) + \sum_d [s^d] \left(g(s)\right) \cdot \tfrac12 T_{d+1}(t) - \sum_d [s^d]\left(g(s)\right) z T_d(t)
     \qquad \!\!\!\!\!\!\!\!\text{(by Fact~\ref{fact:chebid})} \\
  &= \sum_d [s^{d-1}] \left(sg(s)\right) \cdot \tfrac12 T_{d-1}(t) + \sum_d [s^{d+1}] \left(s^{-1}g(s)\right) \cdot \tfrac12 T_{d+1}(t) - \sum_d [s^d] \left(g(s)\right) z T_d(t) \\
  &= \sum_d [s^d]\left( \frac{s}{2} g(s)\right)T_d(t) + \sum_d [s^d] \left(\frac{s^{-1}}{2} g(s)\right)T_d(t) - \sum_d [s^d]\left(z g(s)\right)T_d(t) \\
  &= \sum_d [s^d] \left(\frac{s + s^{-1} - 2z}{2} g(s) \right) T_d(t),
\end{align*}
as desired.
\end{proof}

\subsection{Proof of Claim~\ref{claim:normg}}
\label{sec:normg}

\begin{proof}
By definition of $Z_w$, we have that $z\in [-1,1]$ and thus may set $z=\cos\phi$. We also write  $s=(1+\varepsilon)e^{i\theta}=(1+\varepsilon)\cos\theta +i(1+\varepsilon)\sin\theta$, from which it follows that:
\begin{align*} 
s^2-2sz+1 &= (s-z+\sqrt{z^2-1})(s-z-\sqrt{z^2-1})=(s-\cos\phi+i\sin\phi)(s-\cos\phi-i\sin\phi)\\ 
  &= (s-e^{i\phi})(s-e^{-i\phi})
  = ((1+\varepsilon)e^{i\theta}-e^{i\phi})((1+\varepsilon)e^{i\theta}-e^{-i\phi})  \\
  &= \left((1+\varepsilon)e^{i(\theta+\phi)}-1\right) \left((1+\varepsilon)e^{i(\theta-\phi)}-1\right).
\end{align*}
Recalling that $\delta=\frac{\varepsilon^2}{2(1+\varepsilon)}$, we have that for any $\gamma$, 
\begin{align*} 
|(1+\varepsilon)e^{i\gamma}-1|^2 &= (-1+(1+\varepsilon)\cos\gamma)^2+((1+\varepsilon)\sin\gamma)^2\\
 &=1 - 2(1+\varepsilon)\cos \gamma + (1+\varepsilon)^2\\
 &=2(1+\varepsilon)(1-\cos\gamma+\delta),
\end{align*}

from which it follows that 
\begin{align*} 
|s^2-2sz+1|^2 &= \left|(1+\varepsilon)e^{i(\theta+\phi)}-1\right|^2 \left|(1+\varepsilon)e^{i(\theta-\phi)}-1\right|^2\\
  &= 4(1+\varepsilon)^2(1-\cos(\theta+\phi)+\delta)\cdot (1-\cos(\theta-\phi)+\delta)\\
 &=4(1+\varepsilon)^2(1+\delta)^2\left(1-\frac{\cos(\theta+\phi)}{1+\delta}\right)\left(1-\frac{\cos(\theta-\phi)}{1+\delta}\right)\\
 &= 4(1+\varepsilon)^2(1+\delta)^2\left(\left(1-\frac{\cos\theta\cos\phi}{1+\delta}\right)^2-\left(\frac{\sin\theta\sin\phi}{1+\delta}\right)^2\right)\\
 &= 4(1+\varepsilon)^2(1+\delta)^2\left(\left(1-\frac{z\cos\theta}{1+\delta}\right)^2-\left(\frac{(1-z^2)\sin^2\theta}{(1+\delta)^2}\right)\right)\\
 &= 4(1+\varepsilon)^2\left(\left(1+\delta-z\cos\theta\right)^2-(1-z^2)\sin^2\theta\right)\\
 &= 4(1+\varepsilon)^2\left((1+\delta)^2-2(1+\delta)z\cos\theta-1+z^2+\cos^2\theta\right)\\
 &= 4(1+\varepsilon)^2\left((\cos\theta-(1+\delta)z)^2+(1-z^2)(2\delta+\delta^2)\right).
\end{align*}

Note that the fourth equality uses the sum and difference formulas for sine and cosine. 
\\\\We then have
\begin{align*} 
\left|\frac{s^2-2sz+1}{2}\right|^2 &= (1+\varepsilon)^2\left((\cos\theta-(1+\delta)z)^2+(1-z^2)(2\delta+\delta^2)\right)\\
  &= (1+\varepsilon)^2(1+\delta)^2\left(\left(\frac{\cos\theta}{1+\delta}-z\right)^2+\frac{(1-z^2)(2\delta+\delta^2)}{1+\delta}\right). 
\end{align*}
The claim then follows by multiplicativity of the norm. 
\end{proof}

\subsection{Proofs of Claim~\ref{claim:fupper} and Claim~\ref{claim:bounded}}
\label{sec:fupper}

\paragraph{Proof of Claim~\ref{claim:fupper}}
The objective is to uniformly bound the value of the function
\[ h_\delta(s) = C_w^2 \cdot \prod_{z \in Z_w} h_\delta(s, z), \qquad\text{where}\qquad h_\delta(s, z) = (s - z)^2 + \delta(1 - z^2) \]
for $s \in [-1, 1]$.  When $k, K$ are fixed and $n$ becomes large, all zeros in $Z_w$ approach $-1$ or $+1$, $h_\delta(s, z)$ uniformly approaches $h_0(s, z) = (s - z)^2$, $h_w(s)$ approaches $h_0(s) = p^\infty_w(s)$ and is therefore uniformly bounded.

The main difficulty in extending this argument to finite $n$ is that $h_\delta(s, z)$ can no longer be uniformly bounded by a multiple of $(s - z)^2$ since when $s$ equals $z$, the latter function vanishes but the former one doesn't.  For this reason, we divide the analysis into two parameter regimes.  When $s$ is bounded away from the set of zeros $Z_w$, an approximation of the infinitary term-by-term argument can be carried out.  When $s$ is near the zeroes, we argue that $h_\delta(s)$ cannot be much larger than $h_\delta(s_0)$ for an $s_0$ that is even farther away from $Z_w$, and then argue that $h_0(s_0) = p_w(s_0)^2$ must be small because it represents the square of a probability of a rare event.

\begin{fact}
\label{fact:symm2}
$h_\delta(s, z) h_\delta(s, -z) = h_\delta(-s, z) h_\delta(-s, -z)$.
\end{fact}

\begin{fact}
\label{fact:nonneg}
$h_\delta(s, z) \leq h_\delta(|s|, z)$ when $z \leq 0$ and $s \geq 0$.
\end{fact}

\begin{fact}
\label{fact:quadr}
$h_\delta(s, z) \leq h_\delta(s_0, z)$ when $s_0 \leq s \leq 1$, $s_0 \leq 2z - 1$, and $|z|\leq 1$.
\end{fact}
\begin{proof}
The fact is equivalent to checking that $(s_0-z)^2-(s-z)^2\geq 0$ when $s_0 \leq s \leq 1$ and $s_0 \leq 2z - 1$. If $s\leq z$ then we have that $s_0\leq s\leq z$ from which it immediately follows that $(s_0-z)^2\geq (s-z)^2$. If $s>z$ then $(s-z)^2$ is at most $(1-z)^2$. However, since $|z|\leq 1$, we have that $s_0\leq 2z-1\leq z$ and thus $(s_0-z)^2$ is always at least $(z-(2z-1))^2=(1-z)^2$. Again we have that $(s_0-z)^2\geq (s-z)^2$.
\end{proof}

We begin by reducing to the case of non-negative inputs $s \in [0, 1]$.

\begin{claim}
\label{claim:nonnegative}
Assuming $w \leq K/2$, $h_\delta(s) \leq h_\delta(\abs{s})$.
\end{claim}
\begin{proof}
When $w \leq K/2$ then elements of $Z_w$~\eqref{eq:zeros} can be split into $w$ pairs of the form $A = \{(-1 + 2h/n, 1 - 2h/n)\colon 0 \leq h < w\}$, and $K - 2w$ remaining elements $B = \{-1 + 2h/n\colon w \leq h < K - w\}$ are all
non-positive.  By Fact~\ref{fact:symm2}, $\prod_{(-z, z) \in A} h_\delta(s, z) h_\delta(s, -z) = \prod_{(-z, z) \in A} h_\delta(\abs{s}, z) h_\delta(\abs{s}, -z)$.  By Fact~\ref{fact:nonneg}, $\prod_{z \in B} h_\delta(s, z) \leq \prod_{z \in B} h_\delta(\abs{s}, z)$.  Therefore the product $\prod_{z \in Z_w} h_\delta(s, z) \leq \prod_{z \in Z_w} h_\delta(\abs{s}, z)$.
\end{proof}

The following claim handles values of $s$ in the range $[0, 1 - w/16K]$.

\begin{claim}
\label{claim:ratiobound}
Assuming $0 \leq s \leq 1 - w/16K$,
\[ h_\delta(s, z) \leq 
\begin{cases} (1 + \delta)(s - z)^2, &\text{if $z \leq -1/\sqrt{2}$.} \\
(1 + (64K/w)\delta)(s - z)^2, &\text{if $z \geq 1 - w/32K$}
\end{cases} \]
\end{claim}
\begin{proof}
The ratio $h_\delta(s, z)/(s - z)^2$ equals $1 + ((1 - z^2)/(s - z)^2)\delta$.  The number $(1 - z^2)/(s - z)^2$ is at most $1$ when $s \geq 0$ and $z \leq -1/\sqrt{2}$ and at most the following when $z\geq 1-w/32K$. 
\[ \frac{1 - (1 - w/32K)^2}{((1 - w/16K) - (1 - w/32K))^2} \leq \frac{2w/32K}{(w/32K)^2} = 64K/w. \hfill \]
\end{proof}

\begin{corollary}
\label{cor:ratiobound}
Assuming $0 \leq s \leq 1 - w/16K$ and $n \geq 64K$, $h_\delta(s) \leq e^{65\delta K} h_0(s)$.
\end{corollary}
\begin{proof}
By the choice of parameters, all zeros in $Z_-$ meet the criterion for the first inequality in Claim~\ref{claim:ratiobound}, while all zeros in $Z_+$ meet the criterion for the second one.  Therefore
\begin{align*} 
h_\delta(s) &= C_w^2 \prod_{z \in Z_-} h_\delta(s, z) \prod_{z \in Z_+} h_\delta(s, z) \\
       &\leq C_w^2 \prod_{z \in Z_-} (1 + \delta) (s - z)^2 \prod_{z \in Z_+} (1 + (64K/w)\delta) (s - z)^2  \\
       &\leq (1 + \delta)^{K - w} (1 + (64K/w)\delta)^w \cdot C_w^2 \prod_{z \in Z_-} h_0(s, z) \prod_{z \in Z_+} h_0(s, z) \\
       &\leq e^{\delta K} \cdot e^{64\delta K} \cdot h_0(s). \hfill
\end{align*}
\end{proof}

The following two claims handle values of $s$ in the range $[1 - w/16K, 1]$.

\begin{claim}
\label{claim:plainbound}
Assuming $w \leq K$ and $1 - w/8K \leq s_0 \leq 1 - w/16K \leq s \leq 1$,
\[ h_\delta(s, z) \leq \begin{cases}
h_\delta(s_0, z), &\text{if $z \geq 1 - w/32K$} \\
(1 + w/8K)^2 \cdot h_\delta(s_0, z), &\text{if $z \leq -w/8K$}. \end{cases} \]
\end{claim}
\begin{proof}
By the choice of parameters the first inequality follows from Fact~\ref{fact:quadr}.  For the second one, we upper bound the ratio
\[ \frac{(s - z)^2}{(s_0 - z)^2} \leq \frac{(1 - z)^2}{(1 - z - w/8K)^2} = \biggl(1 + \frac{w/8K}{1 - z - w/8K}\biggr)^2 \leq \biggl(1 + \frac{w}{8K}\biggr)^2. \]
This is greater than one, so $(s - z)^2 + \delta(1 - z^2) \leq (1 + w/8K)^2 ((s_0 - z)^2 + \delta(1 - z^2))$ as desired.
\end{proof}

\begin{corollary}
\label{cor:plainbound}
Assuming $1 - w/8K \leq s_0 \leq 1 - w/16K \leq s \leq 1$ and $n \geq 2K$, $h_\delta(s) \leq e^{w/4} h_\delta(s_0)$.
\end{corollary}

\begin{proof}
By the choice of parameters, all zeros in $Z_-$ meet the criterion for the first inequality in Claim~\ref{claim:plainbound}, while all zeros in $Z_+$ meet the criterion for the second one.  Therefore
\begin{align*} 
h_\delta(s) &= C_w^2 \prod_{z \in Z_-} h_\delta(s, z) \prod_{z \in Z_+} h_\delta(s, z) \\
       &\leq C_w^2 \prod_{z \in Z_-} (1 + w/8K)^2 \cdot h_\delta(s_0, z) \prod_{z \in Z_+}h_\delta(s_0,z)  \\
       &=(1 + w/8K)^{2\abs{Z_-}} \cdot h_\delta(s_0) \\
       &\leq (1 + w/8K)^{2K} \cdot h_\delta(s_0)\leq e^{w/4}h_\delta(s_0). \hfill
\end{align*}
\end{proof}

\begin{claim}
\label{claim:chernoff}
If $s_0$ is of the form $1 - 2h/n$ for some integer $0 \leq h \leq wn/e^2K$ then $0 \leq p_w(s_0) \leq e^{-w}$. 
\end{claim}
\begin{proof}
By Fact~\ref{fact:pw}, $p_w(s_0)$ is the probability that a random string of Hamming weight $h$ and length $n$ has exactly $w$ ones in its first $K$ positions.   The probability that it has at least $w$ ones in its first $K$ positions is at most 
\[ \binom{K}{w} \cdot \frac{h}{n} \cdot \frac{h-1}{n-1} \cdots \frac{h-w+1}{n-w+1} \leq \biggl(\frac{eK}{w}\biggr)^w \biggl(\frac{h}{n}\biggr)^w \leq e^{-w}. \hfill \]
\end{proof}

\begin{proof}[Proof of Claim~\ref{claim:fupper}]
By Claim~\ref{claim:nonnegative} we may assume $s \in [0, 1]$. When $0 \leq s \leq 1 - w/16K$ the result follows from Corollary~\ref{cor:ratiobound}.  When $1 - w/16K \leq \abs{s} \leq 1$, by the assumption $n \geq 64K$ there must exist a value $s_0$ between $1 - w/8K$ and $1 - w/16K$ that is of the form $1 - 2h/n$.  In particular $h \leq wn/e^2 K$.  Then
\[ h_\delta(s) \leq e^{w/4}h_\delta(s_0) \leq e^{w/4}e^{65\delta K} p_w(s_0)^2 \leq e^{65\delta K - 7w/4}, \]
where the inequalities follow from Corollary~\ref{cor:plainbound}, Corollary~\ref{cor:ratiobound}, and Claim~\ref{claim:chernoff}, respectively.
\end{proof}

\paragraph{Proof of Claim~\ref{claim:bounded}}  This proof has a similar structure to that of Claim~\ref{claim:fupper}.  By symmetry we can again restrict attention to inputs $t \in [0, 1]$.  When $t \leq 1 - 2w/n$ then $\abs{p_w(t)}$ is not much larger than $\abs{p_w(t')}$ where $t'$ is the largest number of the form $1 - 2h/n$ not exceeding $t$ for integer $h$.  Otherwise the value $\abs{p_w(t)}$ is not much larger than $\abs{p_w(s_0)}$, for some $s_0 \in [1-w/8K, 1-w/16K]$ of the form $1-2h/n$ for an integer $h$. In turn, $p_w(s_0)$ is the probability of a rare event, so we conclude that $\abs{p_w(t)}$  is small. 

\begin{claim}
\label{claim:apx}
If $-2/n \leq t' \leq t \leq 1 - 2w/n$ then
\[ \abs{t - z} \leq \begin{cases} \abs{t' - z}, &\text{if $z \geq 1 - 2w/n$}, \\
   (1 + 2(t - t')) \abs{t' - z}, &\text{if $z \leq -1/2 -2/n$}. \end{cases} \]
\end{claim}
\begin{proof}
The first part follows because the expressions under the absolute value are nonnegative.  For the second part, we bound the ratio
\[ \frac{t - z}{t' - z} = 1 + \frac{t - t'}{t' - z} \leq 1 + 2(t - t') \]
as desired.
\end{proof} 

\begin{corollary}
\label{cor:apx}
Assuming $n \geq 64K$ and $-2/n \leq t' \leq t \leq 1 - 2w/n$, $\abs{p_w(t)} \leq (1 + 2(t - t'))^K\abs{p_w(t')}$.
\end{corollary}

\begin{proof}
By the choice of parameters, all zeros in $Z_+$ meet the criterion for the first inequality in Claim~\ref{claim:apx}, while all zeros in $Z_-$ meet the criterion for the second one.  Therefore
\begin{align*} 
\abs{p_w(t)} &= C_w \prod_{z \in Z_-} \abs{t-z} \prod_{z \in Z_+} \abs{t-z}\\
       &\leq C_w \prod_{z \in Z_-}  (1 + 2(t - t')) \abs{t' - z}\prod_{z \in Z_+} \abs{t'-z}  \\
       &=(1 + 2(t - t'))^{\abs{Z_-}}\cdot\abs{p_w(t')} \\
       &\leq (1 + 2(t - t'))^{K}\cdot\abs{p_w(t')}. \hfill
\end{align*}
\end{proof}

\begin{proof}[Proof of Claim~\ref{claim:bounded}]
By Fact~\ref{fact:symm} we may assume $w \leq K/2$, and by Claim~\ref{claim:nonnegative} (for $\delta = 0$) we may assume $0 \leq t \leq 1$.  If $t \leq 1 - 2w/n$ then there exists a $t'$ such that $p_w(t')$ is a probability and $0 \leq t - t' \leq 2/n$.  By Corollary~\ref{cor:apx}, $\abs{p_w(t)} \leq (1 + 4/n)^K\abs{p_w(t')} \leq 2\abs{p_w(t')}$.  

If $1 - 2w/n \leq t \leq 1$, then $t \geq 1-w/16K$. 
By the assumption $n \geq 64K$ there must exist a value $s_0$ between $1 - w/8K$ and $1 - w/16K$ that is of the form $1 - 2h/n$.  In particular $h \leq wn/e^2 K$.
By Corollary~\ref{cor:plainbound}, $\abs{p_w(t)} = \sqrt{h_0(t)} \leq e^{w/8} \sqrt{h_0(s_0)} = e^{w/8} \abs{p_w(s_0)}$.  By Claim~\ref{claim:chernoff}, $p_w(s_0)$ is non-negative and at most $e^{-w}$.  Therefore $\abs{p_w(t)} \leq e^{w/8} \cdot e^{-w} \leq 1$.  
\end{proof}

\section{Proofs of Corollary \ref{cor: wtdeglb} and Theorem \ref{thm: symub}}
\label{s:newdetails}
\label{s:symub}

\subsection{Proof of Corollary \ref{cor: wtdeglb}}

\begin{proof}[Proof of Corollary~\ref{cor: wtdeglb}]

Corollary~\ref{cor:symmetricimperfect} implies the existence of a $\phi \left(= \frac{\mu - \nu}{2}\right)$ satisfying $\|\phi\|_1 = 1,~\langle f, \phi \rangle = \eps$ for some $\eps = \Omega(1)$ and $\langle \phi, q \rangle \leq K^{3/2} \cdot 2^{-\Omega\left(\adeg_{1/3}(f)^2/K\right)}$ for any parity of degree at most $K$.

For any $p$ of degree $K$ and weight at most $w$,
\[ \|f - p\|_{\infty} = \|f - p\|_{\infty} \|\phi\|_1 \geq \ip{\phi, f - p} = \ip{\phi, f} - \ip{\phi, p} \geq \eps - w \cdot K^{3/2} \cdot 2^{-\Omega\left(\adeg_{1/3}(f)^2/K\right)}.\]

Thus, we conclude that $W_{\eps/2}(f, K) = K^{-3/2} \cdot 2^{\Omega\left(\adeg_{1/3}(f)^2/K\right)}$.  Corollary~\ref{cor: wtdeglb} now follows using standard error reduction techniques that show that $\adeg_{\eps}(f) = \Theta(\adeg_{1/3}(f))$ for all constants $0 < \eps < 1/2$. 

\end{proof}

\subsection{Proof of Theorem \ref{thm: symub}}
We first require the following lemma. This lemma builds on ideas in \cite[Claim 2]{STT}, which showed a similar result for $t=\Theta(1)$. 

\begin{lemma}\label{lem: andub}
For any $y \in \{0,1\}^n$, denote by $\EQ_y$ the function on $\{0,1\}^n$ that outputs 1 on input $y$, and 0 otherwise.  Then for any $t > 0$ and $d > \sqrt{nt\log n}$, we have $W_{n^{-O(t)}} (\EQ_y, d) \leq 2^{O(nt\log^2(n)/d)}$.
\end{lemma}

\begin{proof}
Note that for any $y \in \pmone^n$, the function $\EQ_y$ is just the $\AND$ function on $n$ input bits (with 0-1 valued output), with possibly negated input variables.  Thus it suffices to give an approximating polynomial for the $\AND$ function on $n$ bits.
We now express $\AND_n$ as $\AND_\ell \circ \AND_{n/\ell}$, where $\ell$ is a parameter we will set later.  We compute the inner $\AND_{n/\ell}$ exactly and approximate the outer $\AND_\ell$ to error $n^{-\Omega(t)}$.  This can be done with a polynomial $p$ of degree $O\left(\sqrt{\ell \log(n^t)}\right)$ \cite{KahnLS96, BCdWZ99}.
Combining the fact that  $p$ is bounded by $1+n^{-\Omega(t)} \leq 2$ at all Boolean inputs with Parseval's identity and the Cauchy-Schwarz inequality, it can be seen that the weight of $p$
is at most $\ell^{O\left(\sqrt{\ell \log(n^t)}\right)}$.\footnote{Building on \cite{BCdWZ99}, It is possible to derive explicit $\eps$-approximating polynomials for $\AND$
where the degree is $O\left(\sqrt{\ell \log(1/\eps)}\right)$  and the weight  is $2^{O\left(\sqrt{\ell \log(1/\eps}\right)}$ rather than 
$\ell^{O\left(\sqrt{\ell \log(1/\eps)}\right)}$. Using this tighter weight bound would  improve our final result by a factor of $\log n$ in the exponent. We omit this tighter result for brevity.} It is well known that the exact multilinear polynomial representation of $\AND_{n/\ell}$ has constant weight.  Hence, by
composing $p$ with the multilinear polynomial that exactly computes  $\AND_{n/\ell}$, we obtain an approximation $q$ for $\AND_n$ of degree $O\left(n\sqrt{\frac{t \log n}{\ell}}\right)$, error $n^{-\Omega(t)}$, and weight $2^{O\left(\sqrt{\ell t\log^3 n}\right)}$.  We now fix the value of $\ell$ to $\ell:=\frac{n^2 t \log n}{d^2} < n$, thereby ensuring that the degree of $q$ is at most $d$.  With this setting of $\ell$, the weight of $q$ is at most $2^{O(nt\log^2(n)/d)}$, proving the lemma.
\end{proof}

\begin{proof}[Proof of Theorem~\ref{thm: symub}]
Let $f : \{0,1\}^n \to \{0, 1\}$ be any symmetric function, corresponding to the univariate predicate $D_f : \bra{0} \cup [n] \to \{0,1\}^n$.
For the purpose of this proof, let us denote by $k_f$ the smallest $i$ for which $f$ is constant on inputs of Hamming weight in the interval $[i+1, n-i-1]$.
Without loss of generality,
$f(x)= 0$ for strings of $x$ Hamming weight between $k_f + 1$ and $n - k_f - 1$.  The case where $f = 1$ on input strings of Hamming weight between $k_f + 1$ and $n - k_f - 1$ can be proved using a similar argument.
Define $\supp(f) := \bra{x \in \{0,1\}^n : f(x) = 1}$. Note that $|\supp(f)| \leq 2 \cdot n^{k_f}$.

Observe that $f(x) = \sum_{y \in \supp(f)}\EQ_y(x)$.
Lemma~\ref{lem: andub} implies, for each $y \in \supp(f)$, the existence of polynomials $p_y$ of degree $K$ and weight $2^{O(nk_f \log^2(n)/K)}$, which approximate $\EQ_y$ to error $\frac{1}{6} \cdot n^{-k_f}$.
Define a polynomial $p \colon \{0, 1\}^n \to \R$ by $p(x) = \sum_{y \in \supp(f)}p_y(x)$.  Clearly $p$ has degree $K$, weight at most $n^{O(k_f)} \cdot 2^{O(nk_f\log^2(n)/K)} = 2^{\tilde{O}(nk_f/K)}$, and error at most $|\supp(f)| \cdot n^{-k_f}/6 \leq 1/3$, where the upper bounds on the weight and error follow from the triangle inequality.

The theorem now follows standard error reduction techniques and Paturi's theorem~\cite{Paturi92}, which states that for symmetric functions, $\adeg(f) = \Theta\left(\sqrt{n \cdot k_f}\right)$.

\end{proof}

\begin{remark}
The upper bound obtained in Theorem~\ref{thm: symub} is more general than as stated,
and the only property of symmetric functions it exploits is that symmetric functions of low approximate 
degree are highly biased. More specifically, 
the proof of Theorem \ref{thm: symub} shows that any function $f \colon \{0, 1\}^n \to \{0, 1\}$ with $\min\{|f^{-1}(0)|, |f^{-1}(1)|\} \leq n^t$ satisfies 
$W_{\eps}(f, K) \leq 2^{\tilde{O}(nt/K)}$ for any $K \geq \sqrt{n t \log n}$. 
\end{remark}

\section{Proof of Theorem~\ref{thm:mainlowerprob}}
\label{s:mainlowerprob}

\noindent \emph{Proof outline.} 
As we explain in more detail in the proof itself, it is sufficient to establish the theorem for fixed $k$ and $K$ and infinitely many $n$ because the statement is downward reducible in $n$.

Using the Chebyshev approximation formulas from Section~\ref{s:apxfromperfect} we derive explicit lower bounds on the large Chebyshev coefficients on the polynomial $p_0$ representing the distinguishing advantage of the AND function on $K$ inputs.  Owing to orthogonality and boundedness of the Chebyshev polynomials, this is a lower bound on the approximate degree of $\AND_K$.  By strong duality as given in the following Claim (see \cite{BIVW}) we obtain Theorem~\ref{thm:mainlowerprob}.

\begin{claim}
\label{claim:duality}
If $\adeg_{\eps/2}(F_n) \geq k$ then there exists a pair of perfectly $k$-wise indistinguishable distributions $\mu$, $\nu$ over $\{0,1\}^n$ such that $\E_{X \sim \mu}[F_n(X)] - \E_{Y \sim \nu}[F_n(Y)] \geq \varepsilon$.
\end{claim}

Recall that the Chebyshev polynomials are orthogonal under the measure $d\sigma(t) = (1 - t^2)^{-1/2} dt$ supported on $[-1, 1]$.  We will need the following identity for their average square magnitude under this measure:
\begin{equation}
\label{eq:cheborth}
\E_{t \sim \sigma}[T_d(t)^2] = 1/2 \qquad\text{when $d > 0$.}
\end{equation}

\begin{proof}[Proof of Theorem~\ref{thm:mainlowerprob}]
By symmetry of the distinguishers, $\mu$ and $\nu$ can be assumed symmetric.  
Let $F_n$ denote the function on $\{0, 1\}^n$ that outputs $\AND_K\left(x|_{\{1, \dots, K\}}\right)$, i.e., $F_n$
outputs the $\AND$ of the first $K< n$ bits of the input.
We prove the theorem for $G_n(x_1, \dots, x_n)=\textsf{NOR}(x|_{\{1, \dots, K\}})$. 
By the symmetry of $0$ and $1$ inputs the theorem also holds for $F_n$.  

First, we claim that the statement of Theorem~\ref{thm:mainlowerprob} is stronger as $n$ becomes larger, so it is sufficient to prove it in the limiting case when $n$ approaches infinity and $k, K$ are fixed.  
Suppose that $\mu$ and $\nu$ are distributions over $n$ bit strings that are 
$k$-wise indistinguishable yet  are $\eps$-reconstructable by $G_n$.  We must show that
there are distributions $\mu'$ and $\nu'$ over $\{0, 1\}^{n-1}$ are 
$k$-wise indistinguishable yet  are $\eps$-reconstructable by $G_{n-1}$.
But this holds for $\mu'$ (respectively $\nu'$) that generate a random sample from $\mu$ (respectively, $\nu$)
and then throw away the last bit. 

If the statement was false then by Claim~\ref{claim:duality} there would exist degree-$k$ polynomials $\tilde{G}_n$ that approximate $G_n$ pointwise on $\B^n$ to within error 
$\eps = \sqrt{2^{-4K + 1} \sum_{d > K} \binom{2K}{K + d}^2}$ for almost all $n$.  Applying the construction from the proof of Fact~\ref{fact:pw} to $\tilde{G}_n$, there exist univariate degree-$k$ polynomials $\tilde{p}_0^n$ approximating $p_0^n$ on the set of points $W_n = \{-1 + 2h/n\colon 0 \leq h \leq n\}$ to within error $\eps$.
We emphasize the dependence on $n$ as it will play a role in the proof.

By Formula~\eqref{eq:zeros} the polynomial $p_0^n$ has the form
\[ p^n_0(t) = C_0^n \prod_{z \in Z_0^n} (t - z), \]
where $Z_0^n = \{-1 + 2h/n\colon 0 \leq h < K\}$ (the set $Z_+$ is empty).  The value $p_n^0(1)$ is the probability that $G_n$ accepts the all-zero string, so it must equal one. The constant $C_0^n$ must therefore equal $\prod_{z \in Z_0^n} (1 - z)^{-1}$.  As $n$ tends to infinity, the set $Z_0$ converges to a single zero at $-1$ of multiplicity $K$, so the sequence $p^n_0$ converges uniformly to the polynomial
\[ p^\infty_0(t) = 2^{-K} (t + 1)^K.  \]
By the triangle inequality, for every $\delta > 0$ and all sufficiently large $n$, $\tilde{p}_0^n$ is within $\eps + \delta$ of $p_0^\infty$ on the set $W_n$.  A degree-$k$ polynomial is determined by its values on $W_{k+1}$ and the set of degree-$k$ polynomials that are within $\eps + \delta$ of $p_0^\infty$ on $W_{k+1}$ is compact.  Therefore the sequence of approximating polynomials $\tilde{p}_0^n$ must contain a subsequence (for values of $n$ that are multiples of $k + 1$) that converges (uniformly) to a limiting degree-$k$ polynomial $\tilde{p}_0^\infty$.  Since $\tilde{p}_0^n$ is within $\eps + \delta$ of $p_0^n$ on $W_n$ for infinitely many $n$, $\tilde{p}_0^\infty$ must be within $\eps + 2\delta$ of $p_0^\infty$ on $W_n$ for infinitely many $n$.  The union of these sets $W_n$ is dense in $[-1, 1]$, and by continuity $p_0^\infty$ can be $\eps + \delta$-approximated by the degree-$k$ polynomial $\tilde{p}_0^\infty$ everywhere on $[-1, 1]$.  As $\delta$ was arbitrary it follows that the $\eps$-approximate degree of $p_0^\infty$ can be at most $k$.

All that remains to prove that this is not true, i.e., to show a lower bound of $k$ on the $\eps$-approximate degree of $p_0^\infty$. 
This lower bound is known (see, e.g., \cite{elkiesmathoverflow}); we provide the details now for completeness.
 Let $q$ be any degree-$k$ polynomial. By Claim~\ref{claim:regcoeffs} the $d$-th Chebyshev coefficient of $p_0^\infty$ equals the $(K + d)$-th regular coefficient of $g^\infty(s) = 2^{-2K} (s + 1)^{2K}$, which has value $2^{-2K} \binom{2K}{K + d}$.  Since $q$ has degree at most $k$, the $d$-th Chebyshev coefficient $c_d$ of $p_0^\infty - q$ must also equal $2^{-2K} \binom{2K}{K + d}$ whenever $\abs{d} > k$.  By symmetry of the Chebyshev coefficients, orthogonality of the Chebyshev polynomials, and Equation~\eqref{eq:cheborth},
 \[
\E_{t \sim \sigma}[(p_0^\infty(t) - q(t))^2] 
  = c_0^2 + \sum\nolimits_{d > 0} (2c_d)^2 \E_{t \sim \sigma}[T_d(t)^2]
  \geq \sum\nolimits_{d > k} 2 \cdot \biggl(2^{-2K} \binom{2K}{K + d}\biggr)^2
  = \eps^2.
\]
It follows that the approximation error $\abs{p^\infty_0(t) - q(t)}$ must exceed $\eps$ for some $t \in [-1, 1]$, contradicting the initial assumption.
\end{proof}

\section{Robustness of Symmetric Secret Sharing Against Consolidation}
\label{sec:robustness}

Consider a secret sharing scheme with $tn$ parties, divided in $n$ blocks of size $t$, that is perfectly secure against size-$k$ coalitions.  If all parties in each block come together and consolidate their information even into a single bit, the number of {\em blocks} against which the scheme remains secure drops to $k/t$.  In general this is the best possible, with linear schemes providing tight examples.

The following corollary shows that if the distribution over shares is symmetric then much better security against this type of attack can be obtained.

\begin{corollary}
\label{cor:appendixd}
Let $f_1, \dots, f_n \colon \B^t \to \B$.  Assume $X, Y$ are $k$-wise indistinguishable symmetrically distributed random variables over $tn$-bit strings.  Write $X = X_1\dots\/X_n$, $Y = Y_1\dots\/Y_n$, where all blocks $X_i$, $Y_i$ have size $t$.  For every $K$, the $n$-bit random variables $X' = f_1(X_1)\dots\/f_n(X_n)$ and $Y' = f_1(Y_1)\dots\/f_n(Y_n)$ are $O((tK)^{3/2} n^K e^{-k^2/1156tK})$-close to being perfectly $K$-wise indistinguishable, assuming $K \leq n/64$.
\end{corollary}

The resulting scheme can be viewed as perfectly secure secret sharing with a potentially faulty dealer:   With probability $1 - p$, the dealer samples perfectly $K$-wise indistinguishable shares 
$X'$ or $Y'$, and with probability $p = O((tK)^{3/2} n^K e^{-k^2/1156tK})$ she leaks arbitrary information about the secret.

For example, if $X, Y$ are visual shares sampled from the dual polynomial~\eqref{eq:dualq} then they are $k = \Omega(\sqrt{tn})$-wise indistinguishable, assuming constant reconstruction error.  Corollary~\ref{cor:appendixd} then says that the induced block-shares $X', Y'$ are $\Omega(\sqrt{n / \log n})$-wise indistinguishable except with probability $\exp -\Omega(\sqrt{n \log n})$.  

If, in addition, $f_1 = \dots = f_n = \AND_t$ then $X', Y'$ are themselves shares of a visual secret sharing scheme that is secure against $\Omega(\sqrt{n / \log n})$-size coalitions.  Therefore symmetric visual secret sharing schemes are downward self-reducible at a small loss in security and dealer error in the following sense:  A scheme for $n$ parties can be derived from one for $tn$ parties by dividing the parties into blocks and $\AND$ing the shares in each block.

\begin{proof}[Proof of Corollary~\ref{cor:appendixd}]
By Theorem~\ref{thm:mainupper}, $X$ and $Y$ are $(tK, O((tK)^{3/2})\cdot e^{-k^2/1156tK})$-wise indistinguishable.  Since any size-$K$ distinguisher against $(X', Y')$ induces a size-$tK$ distinguisher against $(X, Y)$, the former are $(K, \delta = O((tK)^{3/2})\cdot e^{-k^2/1156tK})$-wise indistinguishable.  By Theorem D.1 of~\cite{BIVW}, any pair of $(K, \delta)$-wise indistinguishable distributions over $n$ bits is $2 \delta n^K$-close to a pair of perfectly indistinguishable ones.
\end{proof}

\section{Acknowledgements}
We thank Mark Bun for telling us about the work of Sachdeva and Vishnoi~\cite{sachdeva2013approximation}, and Mert Sa\u{g}lam, Pritish Kamath, Robin Kothari, and Prashant Nalini Vasudevan for helpful comments on a previous version of the manuscript. We are also grateful to Xuangui Huang and Emanuele Viola for sharing the manuscript \cite{hv19}.  Andrej Bogdanov's work was supported by RGC GRF CUHK14207618. Justin Thaler and Nikhil Mande were supported by NSF Grant CCF-1845125.

\bibliography{omnibib}

\appendix
\section{Properties of Symmetric Functions and Distributions}
\label{sec:symmetryapp}
Here, we prove some basic facts that we need about symmetric functions and distributions (see first paragraph of Section~\ref{s:apxfromperfect}). 
Let $Q:\{0,1\}^n\rightarrow\mathbb{R}$ be a function. We say that $Q$ is symmetric if the output of $Q$ depends only on the Hamming weight of its input. If we let $X\colon \{0,1\}^n\rightarrow [0,1]$ denote a probability distribution, we say that $X$ is symmetric if the corresponding function mapping inputs to probabilities is a symmetric function. We need two further facts about such distributions.
\begin{fact}
\label{fact:symmetricmarginals}
Suppose that $X$ is a symmetric distribution over $\{0,1\}^n$. For $S\subseteq \{0,...,n\}$, let $X|_{S}$ denote the projection of $X$ to the indices in $S$. Then, $X|_S$ is also symmetric.
\end{fact}
\begin{proof}
Let $z_w$ be an arbitrary element of $\{0,1\}^{|S|}$ of Hamming weight $w$. Using symmetry of $X$, we can observe that
\[\Pr_X[X|_S=z_w]=\sum_{h=0}^n\sum_{\substack{y\in\{0,1\}^{n-|S|}\\ |y|=h}}\Pr [X|_S=z_w\text{ and }X|_{[n]\setminus S}=y] = \sum_{h=0}^{n-|S|}\binom{n-|S|}{h}\frac{\Pr[|X|=h+w]}{\binom{n}{w+h}}.\] The expression on the right depends only on $w$ and not on $z_w$, so the distribution $X|_S$ must be symmetric also.
\end{proof}
\begin{fact}
\label{fact:symmetrictest}
Suppose that $X$ and $Y$ are symmetric distributions over $\{0,1\}^n$. Then without loss of generality, the best statistical test $Q:\{0,1\}^n\rightarrow [0,1]$ for distinguishing between $X$ and $Y$ is a symmetric function. In particular, we have:
\[\max_{\text{symmetric }Q}\{\E_X[Q(X)]-\E_Y[Q(Y)]\}=\max_{Q}\{\E_X[Q(X)]-\E_Y[Q(Y)]\}.\]
\end{fact}
\begin{proof}
Let $Q^*$ denote $\arg\max_{Q}\{\E_X[Q(X)]-\E_Y[Q(Y)]\}$. If $Q^*$ is symmetric then the proof is complete. If not, define $\tilde{Q}$ as the following symmetrized version of $Q^*$:
\[\tilde{Q}(z):=\E_{\sigma}[Q^*(\sigma (z))],\]
where the expectation is over a uniform permutation $\sigma$. It is clear that $\tilde{Q}$ is a symmetric function and we will write $\tilde{Q}_w$ to denote the value $\tilde{Q}$ takes on any input of Hamming weight $w$. We now show that its distinguishing advantage between $X$ and $Y$ is the same as $Q^*$. Clearly, it is enough to show that $\E_X[\tilde{Q}(X)]=\E_X[Q^*(X)]$ for arbitrary symmetric distribution $X$. This follows from a simple calculation:
\begin{align*}
\E_X[Q^*(X)]
  &= \sum_{w=0}^n\sum_{|x|=w}\Pr [X=x]Q^*(x)=\sum_{w=0}^n\frac{\Pr [|X|=w]}{\binom{n}{w}}\sum_{|x|=w}Q^*(x) \\
  &=\sum_{w=0}^n\Pr [|X|=w]\cdot\tilde{Q}_w=\E_X[\tilde{Q}(X)].
  &
\end{align*}
\end{proof}

\end{document}